\documentclass[
aps,
prx,
superscriptaddress,
amsmath,
amssymb
]
{revtex4-2} 

\usepackage{graphicx}
\usepackage{dcolumn}
\usepackage{bm}
\usepackage[colorlinks, linkcolor=blue]{hyperref}
\usepackage{mathrsfs}
\usepackage{braket}
\usepackage{mathtools}
\usepackage{autobreak}
\usepackage{algorithmic}
\usepackage{algorithm}
\usepackage{color}
\usepackage{amsthm}
\usepackage{quantikz}
\usepackage{empheq}
\usepackage[caption=false]{subfig}
\usepackage{booktabs, multirow}
\allowdisplaybreaks

\newtheorem{thm}{Theorem}

\newtheorem{dfn}{Definition}
\newtheorem{lem}{Lemma}

\newtheorem{problem}{Problem}

\newcommand{\BQP}{\textsf{BQP}}

\newcommand{\proofpara}[1]{\par\medskip\noindent{\bfseries #1.}\par\nobreak\smallskip\noindent\ignorespaces}

\begin{abstract}
The Berry phase is a fundamental geometric quantity for characterizing the geometry and topology of quantum many-body systems.
Previous work established a super-polynomial quantum advantage in Berry phase estimation at inverse-polynomial precision, given an ansatz state approximating the initial ground state.
However, whether this hardness persists for Berry phases quantized by symmetry, which can be distinguished at constant precision, remained open.
In this work, we characterize the computational complexity of quantized Berry phase estimation.
First, we prove that deciding whether the Berry phase is exactly $0$ or $\pi$ is $\mathsf{BQP}$-complete when the spectral gap is inverse-polynomially small.
This result implies that distinguishing the quantized Berry phase still has a super-polynomial quantum advantage, assuming $\mathsf{BPP}\neq \mathsf{BQP}$.
To prove this result, we introduce an encoding that maps the output of any quantum computation to a Berry phase that is exactly $0$ or $\pi$.
Second, we extend this hardness to physically motivated Hamiltonians on a 2D square lattice, including Heisenberg and XY interactions.
Third, we provide a classical polynomial-time algorithm for geometrically local Hamiltonians on fixed-dimensional lattices with constant spectral gaps.
Therefore we identify that the spectral gap is a key resource for the classical tractability, rather than precision.
These results establish a super-polynomial quantum advantage for computing a quantized topological invariant that is unchanged under symmetry- and gap-preserving deformations, and provide new steps towards practical quantum advantage in quantum physics.
\end{abstract}

\begin{document}
\title{
Encoding universal quantum computation into quantized Berry phases: \\
Hardness results and classical algorithms
}

\author{Kazuki Sakamoto}
\email{kazuki.sakamoto.osaka@gmail.com}
\affiliation{%
Graduate School of Engineering Science, The University of Osaka\\
1-3 Machikaneyama, Toyonaka, Osaka 560-8531, Japan.
}

\author{Keisuke Fujii}
\affiliation{%
Graduate School of Engineering Science, The University of Osaka\\
1-3 Machikaneyama, Toyonaka, Osaka 560-8531, Japan.
}
\affiliation{%
Center for Quantum Information and Quantum Biology, The University of Osaka 560-0043, Japan.
}
\affiliation{
Center for Quantum Computing, RIKEN, Hirosawa 2-1, Wako Saitama 351-0198, Japan.
}
\affiliation{
Graduate School of Informatics, Kyoto University, Sakyo-ku, Kyoto, 606-8501, Japan
}

\maketitle


\section{Introduction}\label{sec:introduction}
One of the central goals of modern condensed-matter physics is to identify and characterize the phases of matter that can emerge in quantum many-body systems, and to understand the universal principles that distinguish them. 
Many conventional phases are successfully described within the Landau paradigm, in which phases are characterized by patterns of spontaneous symmetry breaking~\cite{landau2013course}. 
In quantum mechanics, however, there exist phases that cannot be described by this paradigm, including phases with intrinsic topological order and symmetry-protected topological phases~\cite{wen2017colloquium}. 
These developments have revealed the fundamental roles of topology, symmetry, and many-body entanglement in quantum phases~\cite{senthil2015symmetry, chiu2016classification}.
For gapped quantum systems, phases can be understood in terms of continuous deformations of local Hamiltonians.
Ground states connected by a gapped path of local Hamiltonians preserving the relevant symmetries are regarded as belonging to the same phase~\cite{chen2010local, wen2013topological}. 
Topological invariants, which remain unchanged under such deformations, provide robust labels for distinguishing different phases.
This perspective naturally motivates us to study how the ground state evolves as the Hamiltonian is varied along a path.

The Berry phase~\cite{berry1984quantal} is one of the fundamental quantities reflecting such a ground state structure.
When a Hamiltonian is varied adiabatically around a closed loop, a nondegenerate ground state returns to itself up to a phase. 
A part of the accumulated phase depends only on the path traversed by the state, which is the Berry phase.
This geometric quantity and the associated Berry curvature play a central role in the description of quantum matter, underlying phenomena such as electric polarization~\cite{king1993theory, resta1994macroscopic} and quantized Hall conductance~\cite{thouless1982quantized, niu1985quantized}.
Moreover, while a generic Berry phase can take continuous values, suitable symmetries can quantize it to discrete values such as $0$ and $\pi$~\cite{zak1989berry, hatsugai2006quantized}, making it stable under symmetry- and gap-preserving deformations.
It therefore defines a topological invariant of the symmetry-constrained Hamiltonian loop.
For appropriately chosen loops, these invariants provide concrete diagnostics of interacting many-body states.
In spin systems, for example, they reveal local singlet and valence-bond structures that are not captured by conventional symmetry-breaking order parameters~\cite{hatsugai2006quantized, hirano2008topological}.
Beyond one dimension, generalized quantized Berry phases have also been used to identify symmetry-protected topological phases in two-dimensional Heisenberg and XY spin models~\cite{araki2020zq}.
These applications motivate the computation of Berry phases from parametrized Hamiltonians.
Classical numerical methods, including exact diagonalization, tensor-network approaches, and quantum Monte Carlo, have been used to study Berry phases and related topological properties~\cite{hirano2008topological, pollmann2012detection, motoyama2013path}.
Their efficiency, however, depends on the structure of the system and is not guaranteed for general interacting Hamiltonians.

Quantum algorithms offer a promising approach to this computational challenge.
In fact, algorithms based on adiabatic evolution and phase estimation have been proposed to estimate Berry phases~\cite{murta2020berry}.
Also the variational approaches~\cite{tamiya2021calculating} and techniques for suppressing adiabatic phase error~\cite{kiumi2026adiabatic} have been proposed.
These developments strengthen the motivation for studying Berry phase estimation as a candidate for quantum advantage in the characterization of interacting systems.
Nevertheless, the existence of an efficient quantum algorithm does not establish a quantum advantage.
That is, the limitations of particular classical numerical methods do not rule out other efficient classical algorithms.
A complexity-theoretic characterization is therefore essential to determine when Berry phase estimation is classically hard.

Recent work by Hayakawa et al.~\cite{hayakawa2025computational} partially answered this question.
They proved that Berry phase estimation is \BQP-complete for inverse-polynomial precision when a guiding state that has a large overlap with the initial ground state is provided.
This implies that estimating the Berry phase is intractable for classical computers in the worst case under the assumption of $\mathsf{BPP}\neq \BQP$.
However, their result leaves open whether comparable hardness persists for quantized Berry phases. 
Distinguishing zero from $\pi$ requires only constant precision, so hardness in inverse-polynomial precision does not settle this question.
The question is particularly natural in light of dequantization results for the guided local Hamiltonian problem~\cite{gharibian2022dequantizing, gall2024classical}.
That is, when we are given a guiding state that has a large overlap with the ground state and the full Hamiltonian is normalized to have constant norm, the classical algorithm can efficiently estimate the ground state energy with a constant precision, while the problem is $\BQP$-hard when the required precision is inverse polynomially small.
However, constant accuracy in a Berry phase does not automatically place the problem within the scope of such results, since this problem involves the spectral gap and the variation of the Hamiltonian along a closed path as an input.
Relaxing the phase accuracy alone therefore does not determine classical tractability.

In this work, we establish hardness for estimating exactly quantized Berry phases at constant precision and identify conditions under which Berry phases can be estimated classically. 
We consider smooth closed paths of local Hamiltonians $H(\theta)$ on $n$ qubits with a unique ground state throughout the loop. 
Each local interaction term has operator norm at most one and the variation of the path is bounded as $\sup_\theta \|\dot{H}(\theta)\|\leq B_1$.
We also assume access to a guiding state that approximates the ground state of $H(0)$.
For the classical algorithms, this access consists of querying amplitudes and sampling computational-basis strings according to their squared magnitudes.
We have three main results (Fig.~\ref{fig:summary-of-result}):
\begin{enumerate}
    \item We show that deciding whether the Berry phase is zero or $\pi$ is \BQP-hard, even for 5-local Hamiltonians with an inverse-polynomial spectral gap and constant $B_1$ (Theorem~\ref{thm:hardness-BPE}).
    \item We extend the hardness result to Hamiltonians on a 2D square lattice, whose interactions are drawn from any fixed non-2SLD~\footnote{
    The term 2SLD stands for ''the 2-local parts of all interactions in the set are simultaneously locally diagonalizable''.
    The concept was originally introduced in Ref.~\cite{cubitt2016complexity}.
    } interaction set, while allowing an inverse-polynomial spectral gap (Theorem~\ref{thm:hardness-physical-coefficient-variable}).
    The class of non-2SLD interactions consists of interaction sets whose two-qubit parts cannot all be transformed into Ising-type interactions by a single-qubit basis change.
    For example, Heisenberg and XY interactions are included in this class.
    \item We provide an efficient classical algorithm that estimates the Berry phase for geometrically local Hamiltonians on fixed-dimensional lattices under the assumption of constant spectral gap (Theorem~\ref{thm:classical-BPE} and Theorem~\ref{thm:classical-quantized-BPE}).
    The algorithm runs in polynomial time when $B_1$ and the precision are constant, given a sufficiently accurate guiding state.
\end{enumerate}

\begin{figure}
\centerline{
\includegraphics[width=170mm, page=1]{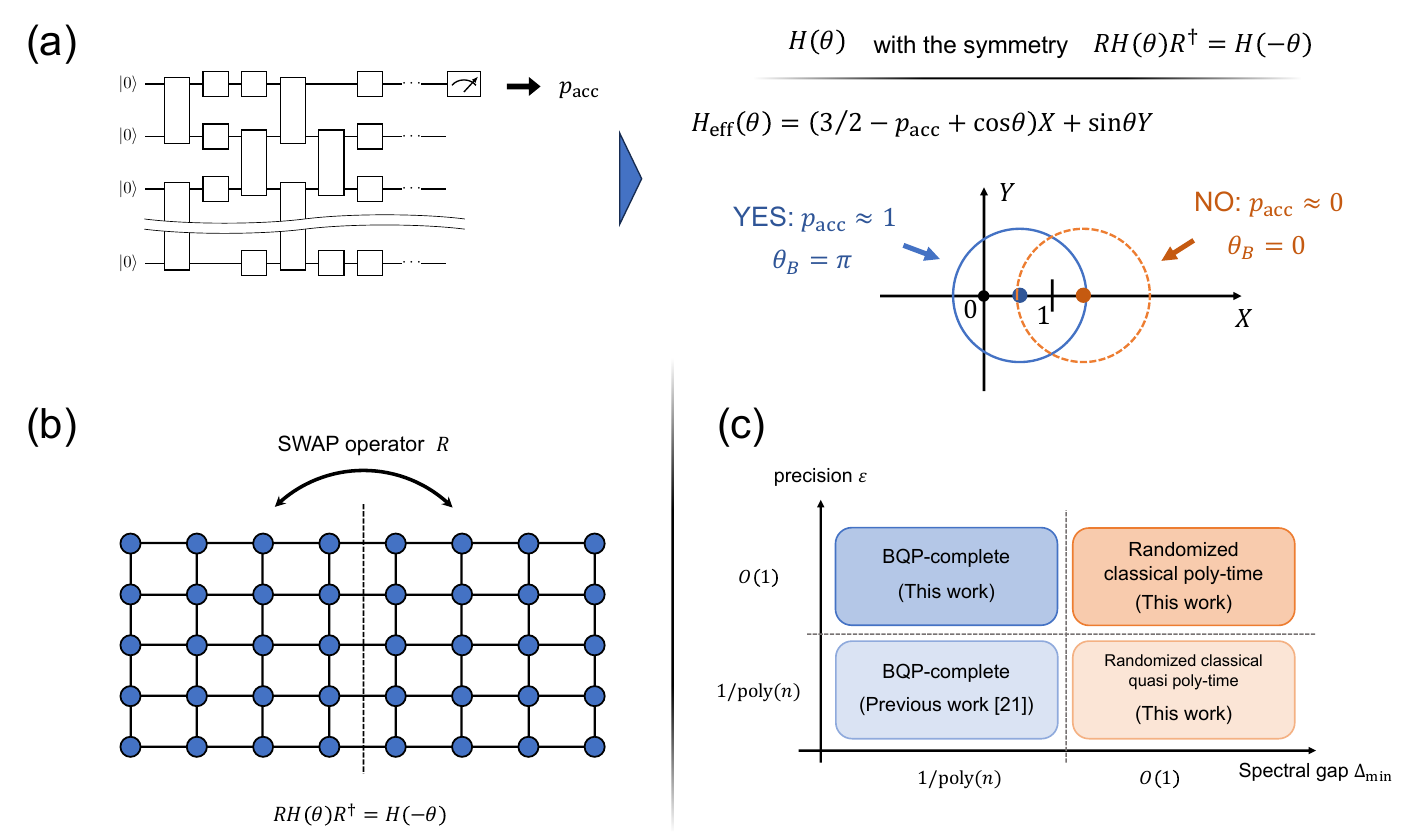}
}
\caption{
Overview of the main results.
(a) Hardness mechanism. 
The acceptance probability \(p_{\mathrm{acc}}\) of the encoded quantum computation is mapped to an effective two-level Hamiltonian in the low-energy subspace. 
In the YES case, \(p_{\mathrm{acc}}\approx 1\) and the loop encloses the origin, giving \(\theta_B=\pi\), whereas in the NO case, \(p_{\mathrm{acc}}\approx 0\) and the loop does not enclose it, giving \(\theta_B=0\).
The symmetry \(R H(\theta)R^\dagger=H(-\theta)\) quantizes the Berry phase of the full Hamiltonian to \(0\) or \(\pi\).
(b) The reflection \(R\) exchanges the two halves of the lattice and realizes the symmetry \(R H(\theta)R^\dagger=H(-\theta)\), which protects the quantization of the Berry phase throughout the reduction.
(c) We assume that $B_1=O(1)$ and $\gamma\geq 1-O(\varepsilon)$.
For inverse-polynomial spectral gaps, Berry-phase estimation is BQP-complete at constant precision for the quantized \(0/\pi\) case as shown in this work. 
For constant-gap geometrically local Hamiltonians, our classical algorithms run in polynomial time at constant precision and quasi-polynomial time at inverse-polynomial precision.
}
\label{fig:summary-of-result}
\end{figure}

The first result implies that the constant precision does not make the problem classically tractable in Berry phase estimation.
This is a contrast with the context of ground state energy estimation.
In the case of estimating the ground state energy of a normalized Hamiltonian with an accurate guiding state, constant precision suffices for an efficient classical algorithm~\cite{gharibian2022dequantizing, gall2024classical}, whereas it does not remove the hardness of estimating a Berry phase.
Moreover, the guiding state in our construction can be chosen as a simple classical subset state, which is a uniform superposition over polynomially many computational basis states, tensored with a fixed single-qubit state.
The hardness holds even if the guiding  state has a squared overlap with the ground state of $H(0)$ at least $1-1/\mathrm{poly}(n)$.
While the simple and high fidelity guiding state was already achieved in Ref.~\cite{hayakawa2025computational}, our result retains these properties while establishing hardness at constant precision.

To establish the hardness of quantized Berry phase estimation, we introduce an encoding that maps the outcome of a quantum computation to a symmetry-quantized binary topological invariant of a Hamiltonian loop.
Previous work~\cite{hayakawa2025computational} constructed a nondegenerate circuit-to-Hamiltonian with a perturbation term that encodes the outcome of computation in a perturbatively small Berry phase.
We instead employ the two-level Bloch Hamiltonian underlying the Su--Schrieffer--Heeger (SSH) model~\cite{su1979solitons}.
This Hamiltonian has an inversion symmetry that quantizes the Berry phase to zero or $\pi$, depending on the relative coupling strengths~\cite{chiu2016classification}. 
We add a single ancilla qubit to a standard circuit-to-Hamiltonian construction~\cite{kitaev2002classical, feynman1986quantum} and consider its two-dimensional ground space. 
By adding a perturbation term conditioned on the final computational time, we induce an effective SSH Hamiltonian within this low-energy subspace.
As illustrated in Fig.~\ref{fig:summary-of-result}(a), the circuit acceptance probability determines whether the corresponding trajectory encloses the origin, yielding a Berry phase of $\pi$ when it does and $0$ when it does not~\cite{zak1989berry}.
Consequently, reducing the perturbation strength changes the energy scale without reducing the phase separation.
Moreover, by ensuring that the full Hamiltonian possesses an exact symmetry, the Berry phase is quantized to either $0$ or $\pi$ for the full Hamiltonian~\cite{zak1989berry, hatsugai2006quantized}, rather than merely for its low-energy effective Hamiltonian.
Exploiting this symmetry together with the analysis of the spectrum by Schrieffer-Wolff transformation~\cite{bravyi2011schrieffer, bravyi2017complexity}, we show that the Berry phase takes different quantized values, $0$ or $\pi$, depending on the outcome of the encoded quantum circuit.

Our construction has a natural implication beyond establishing hardness.
It shows that the outcome of any quantum computation can be encoded in a discrete topological invariant that remains unchanged under continuous deformations preserving the relevant symmetry and keeping the ground-state gap open.
This provides a complementary perspective on the universality of adiabatic quantum computation.
The standard universality result shows that an arbitrary quantum circuit can be simulated by following ground states along a Hamiltonian path~\cite{aharonov2008adiabatic}.
Here, by contrast, the computational output itself is encoded in whether the Hamiltonian loop has Berry phase $0$ or $\pi$.
Thus, determining only this discrete value is already sufficient to solve an arbitrary $\mathsf{BQP}$ decision problem.
This viewpoint suggests a possible route toward robust encodings for adiabatic quantum computation, in which computational information is stored in deformation-invariant quantities such as quantized Berry phases.

Our second result restricts both the interactions and their geometry while preserving this hardness.
We say that the Hamiltonian is non-2SLD when all interactions are drawn from a fixed non-2SLD interaction set and coefficients can be chosen arbitrarily.
In the guided local Hamiltonian problem, $\BQP$-hardness holds even for any non-2SLD Hamiltonian on a 2D square lattice~\cite{cade2022improved}.
Inspired by this result, we establish $\BQP$-hardness of Berry phase estimation for any non-2SLD Hamiltonian on a 2D square lattice \footnote{
Very recently, independent work by Waite~\cite{waite2026computational} established BQP-completeness of guided Berry phase estimation at inverse-polynomial precision for parameterized 2-local qubit Hamiltonians, including weighted Heisenberg interactions on two-dimensional square and triangular lattices. 
Our work instead focuses on symmetry-quantized Berry phases taking exactly \(0\) or \(\pi\), for which we establish hardness already at constant precision.
Moreover, the exact symmetry substantially simplifies the physical reduction by reducing Berry-phase preservation to the preservation of discrete symmetry eigenvalues.
}.
Moreover, the hardness persists even when the parameter-dependent terms are limited to a constant-size region, and also $B_1=O(1)$. 
The guiding state can still be chosen to be a semi-classical encoded state, which admits an efficient classical description and access.

The reduction mainly builds on universal Hamiltonian
simulation~\cite{cubitt2018universal}. 
However, applying a simulation separately at each parameter value is insufficient, because it may not preserve the symmetry or the smoothness of the parameter loop. 
We realize the symmetry as an exchange of the two registers by arranging two copies of the circuit-to-Hamiltonian construction.
We then assign identical perturbative gadgets to terms exchanged by this symmetry, ensuring that the symmetry is preserved throughout the reduction.
Here, perturbative gadgets reproduce the desired interaction within the low-energy subspace using ancilla qubits and carefully chosen Hamiltonian terms.
Also a direct application of the reductions in the previous work encounters a smoothness problem.
That is, when an original coefficient crosses zero, derivatives of the resulting simulator Hamiltonian can diverge.
We avoid this problem by first reducing to a Hamiltonian whose relevant coefficients do not cross zero anywhere along the parameter loop, and then applying the standard perturbative reduction.

The first and second results show that the problem remains $\mathsf{BQP}$-hard even at constant precision as long as the spectral gap is inverse polynomial. 
This naturally raises the question of its complexity when the spectral gap is constant.
Our third result shows that, assuming that both a precision and a spectral gap are constant, the problem can be solved by a randomized classical algorithm in polynomial time when $B_1=O(1)$ and the guiding state has sufficiently large constant fidelity with the ground state. 
Since our hardness results already hold for constant $B_1$ and guiding state with fidelity inverse polynomially close to one, this result implies that the spectral gap is a key factor that can fundamentally change the computational complexity of Berry phase estimation.
When the Berry phase is promised to be either $0$ or $\pi$, the requirement on the guiding state can be further relaxed.
Specifically, when the squared overlap is $\geq 1/2+\Omega(1)$, the algorithm runs in polynomial time, whereas when it is $\geq 1/2$, it runs in quasipolynomial time.

To obtain these results, we first develop a dequantized algorithm for time-dependent Hamiltonian simulation based on Ref.~\cite{gharibian2022dequantizing}.
We then use this algorithm to classically simulate the dynamics generated by quasi-adiabatic continuation.
This dynamics exactly transports the ground state along the adiabatic path and yields the Berry phase~\cite{hastings2010locality}.
The operator norm of the quasi-adiabatic generator is given by $O(B_1/\Delta_{\min})$, which allows the dequantized algorithm to run efficiently as long as $B_1$ and the spectral gap are constant.
When the Berry phase is promised to be $0$ or $\pi$, we further insert a degree-one energy filter, extending the algorithm to the case where the squared overlap of the guiding state is exactly $1/2$. 
Importantly, these algorithms require only that each local term of the Hamiltonian have a constant norm and they do not require the norm of the full Hamiltonian to be constant.
This is because the time evolution is generated not by the Hamiltonian itself, but by the quasi-adiabatic generator, whose operator norm is represented by $O(B_1/\Delta_{\min})$ rather than by the norm of the full Hamiltonian.
In this sense, our setting differs from the dequantization of the guided local Hamiltonian problem~\cite{gharibian2022dequantizing, gall2024classical}, where the full Hamiltonian is normalized to have a constant norm~\footnote{See Refs.~\cite{wu2024classical, zhang2024dequantized} for works considering other normalization and promise.}.

In summary, our results show that the output of any quantum computation can be encoded in a symmetry-quantized topological invariant of a local Hamiltonian loop.
The encoded value is unchanged under continuous deformations that preserve the symmetry and keep the spectral gap open.
Nevertheless, determining this discrete invariant remains $\mathsf{BQP}$-complete when inverse-polynomial spectral gaps are allowed.
We also show that the problem becomes classically tractable in the constant spectral gap regime when both the precision and the derivative of Hamiltonian $B_1$ are constant. 
In establishing the $\mathsf{BQP}$-completeness result, we exploit the topological properties of a simple single-qubit model and introduce reduction techniques that preserve the symmetry. 
These techniques may be of independent interest for future research.
On the algorithmic side, our classical algorithm has a provable complexity and may provide a foundation for efficient Berry-phase estimation in quantum many-body systems.
Our results clarify the potential for super-polynomial quantum advantage in computing the Berry phase, which is one of the fundamental topological invariants in quantum many-body systems, and provide new insight into the computational aspects of topology in quantum many-body physics.

This paper is organized as follows.
In Section~\ref{sec:preliminary}, we briefly introduce the definitions. 
Section~\ref{sec:hardness-constant}  establishes hardness for estimating symmetry-quantized Berry phases, and Section~\ref{sec:hardness-physical} extends the hardness to physical Hamiltonians on a two-dimensional square lattice.
Section~\ref{sec:dequantize} presents our classical algorithms in the constant-gap regime.
Section~\ref{sec:conclusion} concludes with open problems and future directions.

\section{Preliminaries}\label{sec:preliminary}

\subsection{Berry phase}\label{subsec:BP}
Let $H(\theta)$ be a smooth closed Hamiltonian path with a nondegenerate ground state such that $H(0)=H(2\pi)$, and $\ket{\psi_0(\theta)}$ be its ground state. 
Its Berry connection and Berry phase are defined by~\cite{berry1984quantal}
\begin{equation}\label{eq:berry-definition}
 \mathcal A(\theta)=i\braket{\psi_0(\theta)|\dot\psi_0(\theta)},
 \qquad
 \theta_B=\int_0^{2\pi}\mathcal A(\theta)\,\mathrm d\theta
 \pmod{2\pi}.
\end{equation}
The Berry phase is unchanged by a parameter-independent change of basis or by multiplication of the Hamiltonian by a positive scalar. 
Although a generic Berry phase varies continuously with the loop, the symmetry may restrict it to two values. 
Specifically, suppose that a parameter-independent unitary involution $R$ satisfies
\begin{equation}\label{eq:inversion-condition}
 R^2=I,\qquad RH(\theta)R^\dagger=H(-\theta),
\end{equation}
where the parameter is understood modulo $2\pi$. At $\theta=0$ and $\pi$, the unique ground state of $H(\theta)$ is an eigenstate of $R$, with eigenvalue $\xi_0, \xi_\pi \in\{+1,-1\}$. Then it holds that
\begin{equation}\label{eq:berry-parity-formula}
 e^{i\theta_B}=\xi_0\xi_\pi.
\end{equation}
Thus, the Berry phase is determined by the eigenvalues of $R$ at $\theta=0$ and $\pi$~\cite{zak1989berry,hughes2011inversion}. 
For completeness, Appendix~\ref{app:berry-symmetry} derives the identity in the notation used here. 

For angles, we write $|x-y|_{2\pi}=\min_{\ell\in\mathbb Z}|x-y+2\pi\ell|$. The notation $[a,b]_{2\pi}$ denotes the arc traversed from $a$ to $b$ in the increasing angular direction, allowing it to pass through $2\pi$.

\subsection{Schrieffer-Wolff transformation}\label{subsec:SW}
We review the Schrieffer--Wolff transformation following Refs.~\cite{bravyi2011schrieffer,bravyi2017complexity}, which allows us to analyze the low energy spectrum of perturbed Hamiltonians.
Let $H_0$ be an $n$-qubit Hamiltonian with a possibly degenerate ground energy $E_0$, ground space $\mathcal P_0$, and spectral gap $\Delta_0>0$ above that ground space. 
Let $P_0$ denote the projector onto $\mathcal P_0$.
In the context of Hamiltonian complexity, we usually consider the Hamiltonian of the form $H=H_0 + V$ where $\|V\|$ is sufficiently small relative to $\Delta_0$.
If $\|V\|\leq \Delta_0/4$, all eigenvalues of $H_0$ are shifted at most $\Delta_0/4$.
Therefore, the low energy eigenvalues of $H$ lie in an interval of $[E_0-\Delta_0/4, E_0+\Delta_0/4]$ and are still separated from other higher eigenvalues by an amount of $\geq \Delta_0/2$.
Letting $\mathcal{P}$ be that low energy subspace of $H$ and $P$ be a projector onto $\mathcal{P}$, Schrieffer-Wolff transformation provides a unitary $U_{SW}=e^{S}$ that maps $P$ to $P_0$ such that $U_{SW}PU_{SW}^\dagger = P_0$.
Then the effective low energy Hamiltonian of $H$ is given by $H_{SW} = P_0 U_{SW} (H_0+V)U_{SW}^\dagger P_0$ and its perturbative expansion is
\begin{equation*}
    H_{SW} = H_0 P_0 + P_0VP_0 + O\left( \frac{\|V\|^2}{\Delta_0} \right).
\end{equation*}
In this way, we can estimate the low-energy eigenvalues of $H$ with an additive error $O(\|V\|^2/\Delta_0)$ for sufficiently small $\|V\|$.

\subsection{Query-access and sampling-access to vectors and matrices}\label{subsec:query}
We introduce some access model to vectors and matrices, which is used for classical algorithms.
We define query-access to a vector as follows.
\begin{dfn}[Query-access to a vector]\label{dfn:query-vec}
We say that we have query-access to a vector $u\in\mathbb{C}^{N}$ if on input $i\in\{1,\dots,N\}$ we can query the entry $u_i$.
We denote the runtime of implementing one such query as $\bm{q}(u)$.
\end{dfn}
We define query-access to sparse square matrices which allows us to efficiently compute the non-zero entries and their positions.
\begin{dfn}[Query-access to a sparse square matrix]\label{dfn:query-matrix}
We say that we have query-access to an $s$-sparse square matrix $A\in\mathbb{C}^{N\times N}$ if we have queries $\mathcal{Q}^{\mathrm{row}}(A)$ and $\mathcal{Q}^{\mathrm{col}}(A)$ such that:
\begin{itemize}
    \item on input $(i,l)\in \{1,\dots,N\}\times\{1,\dots,s\}$, the query $\mathcal{Q}^{\mathrm{row}}(A)$ outputs the $l$-th non-zero entry of the $i$-th row of A and its column index if this row has at least $l$ non-zero entries, and outputs an error message otherwise;
    \item on input $(j,l)\in \{1,\dots,N\}\times\{1,\dots,s\}$, the query $\mathcal{Q}^{\mathrm{col}}(A)$ outputs the $l$-th non-zero entry of the $j$-th column of A and its row index if this column has at least $l$ non-zero entries, and outputs an error message otherwise.
\end{itemize}
We denote the runtime of implementing one such query as $\bm{q}(A)$.
\end{dfn}
Typically, the cost of implementing one query is polylogarithmic in $N$ when each entry is stored in a (classical) random-access-memory.

Next, we define sampling-access to a vector.
\begin{dfn}[Sampling-access to a vector]\label{dfn:sampling-vec}
We say that we have sampling-access to a vector $u\in\mathbb{C}^{N}$ if we can sample from a probability distribution $p^u$, where $p_i^u=|u_i|^2/\|u\|^2$ for each $i$.
We denote the runtime of generating one sample as $\bm{s}(u)$.
\end{dfn}
Combined with query-access, we define sampling-and-query-access.
\begin{dfn}[Sampling-and-query-access to a vector]\label{dfn:sampling-query-vec}
We say that we have sampling-and-query-access to a vector $u\in\mathbb{C}^{N}$ if we have query-access and sampling-access to $u$ and additionally can get the value $\|u\|$.
We denote the maximum among $\bm{q}(u)$, $\bm{s}(u)$ and the cost of obtaining $\|u\|$ as $\bm{sq}(u)$.
\end{dfn}

These access models are the ones used in the dequantization framework of Ref.~\cite{gharibian2022dequantizing}. 
For time-dependent matrices, the query takes the time parameter as an additional input, and its cost is the maximum over the requested times. 
We assume that the input descriptions allow the coefficients and their first two derivatives to be evaluated efficiently. 

A classical subset state is a normalized uniform superposition over an explicitly listed set of polynomially many computational-basis states. 
It has efficient sampling-and-query-access and can also be prepared by a polynomial-size quantum circuit. 
We also allow tensoring with a constant number of fixed single-qubit states, which does not affect efficient classical sampling-and-query access.
The physical reductions use the corresponding semi-classical encoded states, obtained by applying fixed, constant-size local encodings and adjoining fixed mediator qubits, as in Ref.~\cite{cade2022improved}.

\section{Hardness of quantized Berry phase estimation}\label{sec:hardness-constant}
We begin with a decision problem formulation of Berry phase estimation. 
Here the two allowed arcs are separated by a gap of width $2\varepsilon$.

\begin{problem}[Guided Berry phase estimation]\label{problem:GBPE}
$\mathrm{GBPE}(k,a,b,\varepsilon,\gamma,\Delta_{\min}, B_1, B_2)$ is the following promise problem.
    \label{def:BPE_guided}
    \ \\
    \textbf{Input:}
    A family of $k$-local Hamiltonians $\{H(\theta)=\sum_{i=1}^m h_i(\theta) \}_\theta$ parameterized by $\theta\in [0,2\pi]$ on $n$-qubit {given as a $\mathrm{poly}(n)$-size description of a function of $\theta$}, real numbers $\varepsilon, \Delta_{\mathrm{min}}, \gamma, B_1, B_2 > 0$,
    {$a,b\in [0,2\pi)$} s.t. $2\varepsilon<|b-a|<2\pi-2\varepsilon$, and sampling-and-query-access to a normalized quantum state $\ket{c}$, together with a polynomial-size quantum circuit preparing this state.\\
    \textbf{Promise: }
    \begin{itemize}
        \item ($\mathcal{C}1$): $H(0)=H(2\pi)$, and for all $\theta\in [0,2\pi]$ the following holds: (1) $\|h_i(\theta)\|\leq 1$ for all $i$,  (2) $H(\theta)$ has a non-degenerate ground state, (3) the spectral gap of $\{H(\theta)\}_\theta$ is lower bounded by $\Delta_{\mathrm{min}}$, (4) $\|\dot{H}(\theta)\| \leq B_1$, $\|\ddot{H}(\theta)\| \leq B_2$.
        \item ($\mathcal{C}2$):
            {$|\braket{\psi_0(0)|c}|^2\geq \gamma$}, where $\ket{\psi_0(0)}$ is a normalized ground state of $H(0)$.
        \item
              The Berry phase $\theta_B =
                  \int_{0}^{2\pi}
                  i\bra{\psi_0(\theta)}\frac{\mathrm{d}}{\mathrm{d}\theta}\ket{\psi_0(\theta)}
                  \mathrm{d}\theta$ lies either in  $[a+\varepsilon,b-\varepsilon]_{2\pi}$ or in $[b+\varepsilon,a-\varepsilon]_{2\pi}$.
    \end{itemize}
    \textbf{Output:}
    {1 if $\theta_B\in[a+\varepsilon,b-\varepsilon]_{2\pi}$ and 0 if $\theta_B\in[b+\varepsilon,a-\varepsilon]_{2\pi}$.}
\end{problem}

In Ref.~\cite{hayakawa2025computational}, they proved that this problem is $\BQP$-hard for the precision $\varepsilon=\Theta(1/\mathrm{poly}(n))$ with appropriately chosen other parameters.
Next we prove that the problem is $\BQP$-hard even for $\varepsilon=\Theta(1)$.

\begin{thm}\label{thm:hardness-BPE}
    The problem $\mathrm{GBPE}(k,a,b,\varepsilon,\gamma,\Delta_{\min}, B_1, B_2)$ is $\BQP$-hard for $k\geq 5$, any $0<\varepsilon \leq \pi/4$ and $\Delta_{\min} = \Theta(1/\mathrm{poly}(n))$.
    The hardness holds even for $B_1, B_2 =O(1)$, a guiding state $\ket{u}$ of a classical subset state tensored with a fixed single-qubit state satisfying $|\braket{\psi_0(0)|u}|^2\geq 1-1/\mathrm{poly}(n)$, and the Berry phase promised to be $\theta_B=0$ or $\pi$.
\end{thm}
\begin{proof}
    We begin by recalling the encoding illustrated in Fig.~\ref{fig:summary-of-result}(a). 
    The circuit acceptance probability determines the hopping amplitude of an effective SSH Hamiltonian, so that the corresponding trajectory encloses the origin for YES instances and does not for NO instances, yielding Berry phases \(\pi\) and \(0\), respectively. 
    The proof below makes this effective picture rigorous for the full Hamiltonian. 
    We first establish the low-energy effective Hamiltonian and spectral gap using the Schrieffer--Wolff transformation.
    Then we use an exact symmetry to guarantee the quantization of the full Berry phase. 
    Finally, we construct a simple guiding state with large overlap with the initial ground state.
    
    \proofpara{Construction of the parametrized Hamiltonian}
    Let $x$ be an instance of a \BQP\ problem.
    Let \(U_x=U_TU_{T-1}\cdots U_1\underbrace{I\cdots I}_{M}\) be a
    quantum circuit with pre-amplified success probability and \(M\)-step initial idling
    on \(n\) qubits such that
    \begin{itemize}
        \item If $x \in L_{\mathrm{yes}}$, then $p_1= \| \Pi_1 U_x \ket{0^n} \|^2 = 1 - O(2^{-n})$,
        \item If $x \in L_{\mathrm{no}}$, then $p_1  =O(2^{-n})$.
    \end{itemize}
    The $n$-qubit register is initialized to $\ket{0^n}$.
    Here, $M$ is chosen such that $T/(T+M+1)\leq 1/\mathrm{poly}(n)$.

    Let us consider a standard 5-local circuit-to-Hamiltonian construction \cite{kitaev2002classical} without the output term $H_{\mathrm{hist}}=H_\text{in}+H_\text{prop}+H_\text{clock}$ on $\mathcal{H}_n\otimes \mathcal{H}_\text{clock}$ where
    \begin{align*}
         & H_{\text{in}} \coloneqq \sum_i \ket{1}\bra{1}_{A_i}  \otimes \ket{0}\bra{0}_{C_1}                  \\
         & H_{\text{clock}} \coloneqq \sum_{t=1}^{T+M-1} \ket{0}\bra{0}_{C_t}\otimes \ket{1}\bra{1}_{C_{t+1}} \\
         & H_{\text{prop}} \coloneqq
        \sum_{t=1}^{T+M} H_t
    \end{align*}
    where
    $$H_t = -\frac{1}{2}U_t\otimes\ket{t}\bra{t-1}_C-\frac{1}{2} U_{t}^\dagger \otimes \ket{t-1}\bra{t}_C +\frac{1}{2}(\ket{t}\bra{t}_C+\ket{t-1}\bra{t-1}_C). $$
    Here, $\mathcal{H}_n$ is the $n$-qubit Hilbert space and $\mathcal{H}_\text{clock}$ is the clock space where time is represented by unary, and $C_i$ represents the $i$-th qubit in the clock register.
    The expression for $H_t$ is its restriction to the legal clock subspace, where $\ket t=\ket{1^t0^{T+M-t}}$, as used in Ref.~\cite{kitaev2002classical}.
    Then the non-degenerate ground state of $H_{\mathrm{hist}}$ is given by
    $$
        \ket{\psi_{\text{hist}}}\coloneqq
        \frac{1}{\sqrt{T+M+1}}
        \sum_{t=0}^{T+M}
        U_tU_{t-1}...U_1 \ket{0^n}\otimes \ket{t},
    $$
    where 
    $U_t=I_{n}$ for $1\leq t\leq M$, and \(U_{M+j}\) is the \(j\)-th gate of the original amplified computation for \(1\leq j\leq T\).
    Let \(N\coloneqq T+M+1\).
    Note that $H_{\mathrm{hist}}$ has a spectral gap of $\Delta(H_\text{hist})=\Omega\left(\frac{1}{N^3}\right)$ \cite{gharibian2012hardness}.

    Adding one ancilla qubit, we let $H(\theta)$  be a Hamiltonian on $\mathcal{H}_n\otimes \mathcal{H}_a\otimes \mathcal{H}_\text{clock}$ defined as
    \begin{equation}\label{eq:hamiltonian-5-local}
        H(\theta) \coloneqq H_{\mathrm{hist}} + r V(\theta)
    \end{equation}
    where \(r\) is an inverse-polynomial parameter chosen below, and
    \begin{equation}\label{eq:perturbation}
        V(\theta):=
        \left[ \left(
        \frac{3}{2} - \Pi_1 + \cos{\theta}\right) \otimes X_a + \sin{\theta} Y_a 
        \right]
        \otimes \ket{1}\bra{1}_{C_{T+M}},
    \end{equation}
    where $\Pi_1 = \ket{1}\bra{1}_\text{out}$.
    Here $\mathcal{H}_a$ is a single qubit Hilbert space for the added ancilla qubit, and \(C_{T+M}\) is the final clock qubit. On the legal clock subspace, \(\ket{1}\bra{1}_{C_{T+M}}\) coincides with the projector onto the final clock state \(\ket{T+M}\). 
    It can be seen that
        {$H(0)=H(2\pi)$ and}
    $$\left\|\frac{\mathrm{d}H(\theta)}{\mathrm{d}\theta}\right\|\in \mathcal{O}(1) 
    \text{ \ and \ } 
    \left\|\frac{\mathrm{d}^2H(\theta)}{\mathrm{d}\theta^2}\right\|\in \mathcal{O}(1)$$
    for all $\theta\in  \mathbb{R}$.

    The restriction of $V(\theta)$ to the ground space of $H_\text{hist}$ yields a standard SSH model.
    That is, let $P_0=\ket{\psi_\text{hist}}\bra{\psi_\text{hist}}\otimes I_a$ be a projector onto the ground space of $H_\text{hist}$, then we obtain 
    \begin{equation*}
        P_0 V(\theta) P_0 
        = \ket{\psi_\text{hist}}\bra{\psi_\text{hist}} \otimes 
        V^{(1)}(\theta)
    \end{equation*}
    where
    \begin{equation*}
        V^{(1)}(\theta) 
        \coloneqq 
        \frac{1}{N} 
        \left[ \left(
        v_\text{eff}(p_1) + \cos{\theta}\right) X_a + \sin{\theta} Y_a 
        \right], 
        \quad 
        v_\text{eff}(p_1)\coloneqq \frac{3}{2} - p_1.
    \end{equation*}
    This is exactly the Hamiltonian of SSH model and Berry phase of $V^{(1)}(\theta)$ is 0 (resp. $\pi$) if $v_\text{eff}(p_1) > 1$ (resp. $v_\text{eff}(p_1)<1$).
    By the definition of \BQP circuit, we have $v_\text{eff}(p_1) \leq \frac{1}{2}+O(2^{-n})$ for the YES instance and $v_\text{eff}(p_1) \geq \frac{3}{2}-O(2^{-n})$ for the NO instance.
    We will use these facts to analyze the Berry phase of our parametrized Hamiltonian $H(\theta)$ later.

    \proofpara{Analysis of the spectral gap}
    Now we show that the Hamiltonian $H(\theta) = H_{\mathrm{hist}} + r V(\theta)$ has a unique ground state and has a spectral gap of $\Delta(H(\theta))\geq 1/\mathrm{poly}(n)$ for all $\theta$.
    By direct calculation, we can easily confirm that $\|rV(\theta)\|\leq \frac{5}{2}r$.
    Thus, as far as $\frac{5}{2}r \leq \frac{\Delta(H_\text{hist})}{4}$, the ground energy and the first excited energy are separated from the second excited energy by $\Delta(H_\text{hist})/2\geq 1/\mathrm{poly}(n)$.

    Next we analyze the ground energy and the first excited energy of $H(\theta)$ by Schrieffer–Wolff transformation.
    As reviewed in Sec.~\ref{subsec:SW}, the effective low energy Hamiltonian is given by
    \begin{equation*}
        P_0 U_{SW} (H_\text{hist}+rV(\theta))U_{SW}^\dagger P_0 
        = rP_0V(\theta)P_0 + O\left( \frac{r^2}{\Delta(H_\text{hist})} \right).
    \end{equation*}
    Recall that 
    \begin{equation*}\label{eq:first-order-SSH}
        P_0 V(\theta) P_0 
        = \ket{\psi_\text{hist}}\bra{\psi_\text{hist}} \otimes 
        V^{(1)}(\theta), \quad
        V^{(1)}(\theta) 
        \coloneqq 
        \frac{1}{N} 
        \left[ \left(
        v_\text{eff}(p_1) + \cos{\theta}\right) X_a + \sin{\theta} Y_a 
        \right], 
        \quad 
        v_\text{eff}(p_1)\coloneqq \frac{3}{2} - p_1.
    \end{equation*}
    Two eigenvalues of $V^{(1)}(\theta)$ are calculated as
    \begin{equation*}
        E_{\pm}^{(1)} 
        = \pm \frac{1}{N} | v_\text{eff}(p_1) + e^{i\theta} | 
        = \pm \frac{1}{N} \sqrt{v_\text{eff}^2 +1 +2v_\text{eff}\cos{\theta}} .
    \end{equation*}
    Since $0\leq p_1\leq 1$, $1/2\leq v_\text{eff}(p_1)\leq 3/2$, the minimum value of $\sqrt{v_\text{eff}^2 +1 +2v_\text{eff}\cos{\theta}}$ is achieved with $e^{i\theta}=-1$.
    Recalling that we have the promise that $p_1 \leq O(2^{-n})$ or $p_1 \geq 1-O(2^{-n})$, we have $\frac{1}{3}\leq \frac{1}{2}-O(2^{-n}) \leq \sqrt{v_\text{eff}^2 +1 +2v_\text{eff}\cos{\theta}} \leq \frac{5}{2}$
    and we obtain
    \begin{equation*}
        \frac{1}{3N} \leq E_+^{(1)}\leq \frac{5}{2N}, \quad
        -\frac{5}{2N} \leq E_-^{(1)}\leq -\frac{1}{3N},
    \end{equation*}
    for all $\theta$.
    Combined with the argument of Schrieffer–Wolff transformation, the two lowest energies of $H(\theta)$ are
    \begin{equation*}
        -\frac{5r}{2N} - O(N^3r^2) \leq E_0 \leq -\frac{r}{3N} + O(N^3r^2),
    \end{equation*}
    \begin{equation*}
        \frac{r}{3N} - O(N^3r^2) \leq E_1 \leq \frac{5r}{2N} + O(N^3r^2)
    \end{equation*}
    As a result, the ground space of $H(\theta)$ is non-degenerate and, as far as $r=o(1/N^4)$ its spectral gap is lower bounded by $\Delta_{\min} = \frac{2r}{3N} - O(N^3r^2) \geq 1/\mathrm{poly}(n)$.
    Note that $r=o(1/N^4)$ is sufficiently small with respect to $\Delta(H_\mathrm{hist})/10$, so the first excited energy and the second excited energy are still well separated by the amount of $\Delta(H_\text{hist})/2\geq 1/\mathrm{poly}(n)$.

    \proofpara{Analysis of the Berry phase}
    We consider the unitary $U_{\text{sym}}=I\otimes X_a$, then the parametrized Hamiltonian has a symmetry condition
    \begin{equation*}
        U_{\text{sym}}H(\theta)U_{\text{sym}}^\dagger = H(-\theta).
    \end{equation*}
    As discussed in Section~\ref{subsec:BP} and Appendix~\ref{app:berry-symmetry}, the unique ground states of $H(0)$ and $H(\pi)$ are also eigenstates of $U_{\text{sym}}$.
    That is,
    \begin{equation*}
        U_{\text{sym}}\ket{\psi_0(0)}
        = \xi_0 \ket{\psi_0(0)}, \quad 
        U_{\text{sym}}\ket{\psi_0(\pi)}
        = \xi_\pi \ket{\psi_0(\pi)},
    \end{equation*}
    where $\ket{\psi_0(\theta)}$ is a unique ground state of $H(\theta)$.
    Then we can calculate the Berry phase as 
    \begin{equation}\label{eq:berry-sym}
         e^{i\theta_B}=\xi_0\xi_\pi,
    \end{equation}
    as shown in Appendix~\ref{app:berry-symmetry}.
    Therefore we only have to compute the eigenvalues $\xi_0$ and $\xi_\pi$ of $U_{\text{sym}}$ for $\ket{\psi_0(0)}$ and $\ket{\psi_0(\pi)}$.

    To obtain $\xi_0$ and $\xi_\pi$, we use Schrieffer–Wolff transformation again.
    Since $U_{\text{sym}}=I\otimes X_a$ commutes with $H(0)$ and $H(\pi)$, the ground states $\ket{\psi_0(0)}$ and $\ket{\psi_0(\pi)}$ belong to either $+1$ eigenspace of $I\otimes X_a$ or $-1$ eigenspace of $I\otimes X_a$.
    In the $+1$ eigenspace of $I\otimes X_a$, we can write
    \begin{equation*}
        H(0) = H_\text{hist} + r\left( \frac{5}{2} - \Pi_1 \right) \otimes \ket{1}\bra{1}_{C_{T+M}}, 
        \quad 
        H(\pi) = H_\text{hist} + r\left( \frac{1}{2} - \Pi_1 \right) \otimes \ket{1}\bra{1}_{C_{T+M}}.
    \end{equation*}
    Using the first-order non-degenerate perturbation theory, we get the ground energy
    \begin{equation*}
        E_{0,+}(0) = \frac{r}{N}\left( \frac{5}{2} - p_1 \right) + O\left(\frac{r^2}{\Delta(H_\text{hist})}\right),
        \quad 
        E_{0,+}(\pi) = \frac{r}{N}\left( \frac{1}{2} - p_1 \right) + O\left(\frac{r^2}{\Delta(H_\text{hist})}\right).
    \end{equation*}
    Similarly, we can calculate the ground energy in the $-1$ eigenspace of $I\otimes X_a$ as
    \begin{equation*}
        E_{0,-}(0) = - \frac{r}{N}\left( \frac{5}{2} - p_1 \right) + O\left(\frac{r^2}{\Delta(H_\text{hist})}\right),
        \quad 
        E_{0,-}(\pi) = - \frac{r}{N}\left( \frac{1}{2} - p_1 \right) + O\left(\frac{r^2}{\Delta(H_\text{hist})}\right).
    \end{equation*}
    Recalling that $0 \leq p_1 \leq 1$, the ground state of $H(0)$ always belongs to the $-1$ eigenspace of $I\otimes X_a$, and then $\xi_0=-1$.
    On the other hand, the ground state of $H(\pi)$ belongs to the $+1$ ($-1$ resp.) eigenspace of $I\otimes X_a$ in the YES (NO resp.) case.
    Therefore
    \begin{equation*}
        \xi_\pi = 
        \begin{cases}
            +1 & \text{if } x\in L_\text{yes} \\
            -1 & \text{if } x\in L_\text{no}.
        \end{cases}
    \end{equation*}
    It follows from Eq.~\eqref{eq:berry-sym} that
    \begin{equation*}
        \theta_B = 
        \begin{cases}
            \pi & \text{if } x\in L_\text{yes} \\
            0 & \text{if } x\in L_\text{no}.
        \end{cases}
    \end{equation*}

    \proofpara{Guiding state}
    Let 
    \begin{equation*}
        \ket{c_\text{hist}} \coloneqq \frac{1}{\sqrt{M+1}}\sum_{t=0}^M \ket{0^n}\otimes \ket{t}.
    \end{equation*}
    The initial idling generates a large overlap with the history state:
    \begin{equation*}
        |\braket{\psi_\text{hist}|c_\text{hist}}|^2\geq \frac{M+1}{T+M+1}=1-\frac{T}{T+M+1}\geq 1-1/\mathrm{poly}(n)
    \end{equation*}
    As shown before, the ground state of $H(0)$ lies in the $-1$ eigenspace of $X_a$.
    Therefore we choose
    \begin{equation*}
        \ket{c} = \ket{c_\text{hist}}\otimes \ket{-} _a.
    \end{equation*}
    By the first-order perturbation theory, the ground state $\ket{\psi_0(0)}$ of $H(0)$ has a large overlap with $\ket{\psi_\text{hist}}\otimes \ket{-} _a$:
    \begin{equation*}
        |\bra{\psi_0(0)}(\ket{\psi_\text{hist}}\ket{-} _a)|^2 \geq 1- O(r/\Delta(H_\text{hist})) \geq 1 - 1/\mathrm{poly}(n),
    \end{equation*}
    where the last inequality holds since we have chosen $r$ such that $r=o(1/N^4)=o(\Delta(H_\text{hist}))$.
    Therefore we obtain $|\braket{\psi_0(0)|c}|^2 \geq 1 - 1/\mathrm{poly}(n)$.
    Sampling-and-query-access to $\ket{c}$ can be constructed classically in polynomial time since $\ket{c}$ only has a support on polynomially many number of computational basis states, then we finish the proof.
\end{proof}

The quantum algorithm of Ref.~\cite{hayakawa2025computational} (Theorem~2 and Section~4) gives containment in $\BQP$ when $1/\varepsilon$, $1/\Delta_{\min}$, $1/\gamma$, $B_1$ and $B_2$ are polynomially bounded, using the given preparation circuit for the guiding state. 
Combined with Theorem~\ref{thm:hardness-BPE}, this gives $\BQP$-completeness in that regime.

\proofpara{Implications for adiabatic quantum computation}
The construction in Theorem~\ref{thm:hardness-BPE} admits a direct interpretation as an adiabatic computation in which the computational output is encoded in a quantized Berry phase. 
The guiding state $\ket{c}$ constructed above can be prepared efficiently and has squared overlap at least $1-1/\mathrm{poly}(n)$ with the ground state of $H(0)$.
Starting from this state, one can adiabatically evolve the system from $\theta = 0$ to $\theta = \pi$ in polynomial time since $\Delta_{\min}^{-1}$, $B_1$, and $B_2$ are polynomially bounded~\cite{jansen2007bounds}.
The ground state of $H(\pi)$ is an eigenstate of $X_a$ with eigenvalue $+1$ for a YES instance and $-1$ for a NO instance.
Therefore, measuring $X_a$ distinguishes the two cases with bounded error even in the presence of initial-state and adiabatic errors.
Since $\xi_0=-1$, Eq.~\eqref{eq:berry-sym} gives $e^{i\theta_B}=-\xi_\pi$, and this measurement reads out the binary invariant of the Hamiltonian loop.

This provides a complementary perspective on the standard universality result of adiabatic quantum computation~\cite{aharonov2008adiabatic}.
In the standard setting, the output is encoded in the ground state of the final Hamiltonian, which can change under gap-preserving deformations of the Hamiltonian path.
In our construction, by contrast, the output is encoded in a symmetry-quantized topological invariant.
The encoded answer remains unchanged under continuous deformations of the loop that preserve $X_aH(\theta)X_a=H(-\theta)$ and keep the ground state nondegenerate and gapped throughout the loop.
Indeed, the ground states at \(\theta=0,\pi\) vary continuously under such a deformation, while their eigenvalues of \(X_a\) are restricted to $\pm 1$. 
Thus \(\xi_0\) and \(\xi_\pi\) cannot change, and Eq.~\eqref{eq:berry-sym} fixes the Berry phase.
This invariance motivates the use of quantized Berry phases as robust encodings of computational outputs in adiabatic quantum computation.

\section{Improving hardness for physical Hamiltonians}\label{sec:hardness-physical}

We now strengthen the hardness result of Theorem~\ref{thm:hardness-BPE} to physically motivated Hamiltonians.
Specifically, we show that the hardness persists for any fixed non-2SLD Hamiltonian on a 2D square lattice.  

Our proof is mainly based on Refs.~\cite{cubitt2018universal, zhou2026universal, cade2022improved}.
However, we have to preserve the symmetry to quantize the Berry phase and to ensure that the Hamiltonian is smooth throughout the entire parameter region.
To this end, it is insufficient to simulate the spectrum at each parameter point and it requires additional care.

A fixed finite set $\mathcal S$ of Hermitian interactions on at most two qubits is \emph{non-2SLD} if there is no single-qubit unitary $U$ such that every two-qubit $S\in\mathcal S$ has the form
\begin{equation*}
 (U\otimes U)S(U^\dagger\otimes U^\dagger)
 =\alpha_S Z\otimes Z+A_S\otimes I+I\otimes B_S.
\end{equation*}
An $\mathcal S$-Hamiltonian is a local Hamiltonian each of whose terms is drawn from $\mathcal{S}$ with a real coefficient.
Coefficients of either sign are allowed, and thus the theorem does not impose an antiferromagnetic restriction.
On the square lattice, two-qubit interactions act only on nearest-neighbor sites.

\begin{thm}\label{thm:hardness-physical-coefficient-variable}
    For every fixed non-2SLD interaction set $\mathcal{S}$, the restriction of guided Berry phase estimation to $\mathcal{S}$-Hamiltonians on a 2D square lattice is $\BQP$-hard with any $0<\varepsilon \leq \pi/4$ and $\Delta_{\min}=\Theta(1/\mathrm{poly}(n))$.
    The hardness holds even for $B_1, B_2 =O(1)$, a guiding state $\ket{u}$ of a semi-classical encoded state satisfying $|\braket{\psi_0(0)|u}|^2\geq 1-1/\mathrm{poly}(n)$, and the Berry phase promised to be $\theta_B=0$ or $\pi$.
\end{thm}
\begin{proof}
    \proofpara{Construction and analysis of the real and spatially sparse parametrized Hamiltonian}
    We first modify the 5-local Hamiltonian in Eq.~\eqref{eq:hamiltonian-5-local} to be real in the computational basis and spatially sparse.
    Let $x$ be an instance of a \BQP\ problem.
    Let \(U_x=U_TU_{T-1}\cdots U_1\) be a
    quantum circuit with pre-amplified success probability on \(n\) qubits such that
    \begin{itemize}
        \item If $x \in L_{\mathrm{yes}}$, then $p_1= \| \Pi_1 U_x \ket{0^n} \|^2 = 1 - O(2^{-n})$,
        \item If $x \in L_{\mathrm{no}}$, then $p_1  =O(2^{-n})$.
    \end{itemize}
    Without loss of generality, we take the encoded quantum circuit to be real, using a universal gate set of $\{H,\mathrm{Toffoli}\}$~\cite{shi2002both, aharonov2003simple}.
    Note that we do not add the initial idling in this step.
    
    Next, using standard SWAP-based sparsification~\cite{oliveira2005complexity, zhou2026universal}, where the computational information is moved onto the next fresh qubits by SWAP operators after each single gate operation, we transform the circuit into a real spatially sparse circuit $U_x^{\mathbb R,\mathrm{sp}} = U_{T_{\mathrm{sp}}}\cdots U_1$ on $n_{\mathrm{sp}}$ qubits.
    Here $T_{\mathrm{sp}}=\mathrm{poly}(n)$, $n_{\mathrm{sp}}=\mathrm{poly}(n)$, and each qubit participates in only a constant number of nontrivial gates.
    By placing each clock qubit near the corresponding qubit for computation, the resulting circuit-to-Hamiltonian construction can be arranged so that its interaction graph is spatially sparse and 6-local.
    Note that penalty terms for the initial state can also be chosen so that the Hamiltonian is spatially sparse as shown in Ref.~\cite{oliveira2005complexity} by checking each computational qubit just before applying computational gates.

    We now choose the length $M$ of the initial idling step.
    The initial idling does not require introducing any additional computational qubits, and it suffices to append a one-dimensional chain consisting only of clock qubits before the computation begins.
    We choose $M$ such that
    \begin{equation*}
    \frac{T_{\mathrm{sp}}}
    {T_{\mathrm{sp}}+M+1}
    \le
    \frac{1}{\operatorname{poly}(n)}.
    \end{equation*}
    The resulting Hamiltonian $H_\text{hist}^{\mathbb{R},\mathrm{sp}}$ therefore remains spatially sparse and 6-local.
    The ground state of $H_\text{hist}^{\mathbb{R},\mathrm{sp}}$ is the corresponding history state
    \begin{equation*}
        \ket{\psi_\text{hist}}
        =\frac{1}{\sqrt{T_{\mathrm{sp}}+M+1}} \sum_{t=0}^{T_{\mathrm{sp}}+M} U_t U_{t-1}\dots U_1 \ket{0^{n_{\mathrm{sp}}}}\otimes \ket{t},
    \end{equation*}
    and, by the same argument as Ref.~\cite{oliveira2005complexity}, the spectral gap of the Hamiltonian is $\Delta(H_\text{hist}^{\mathbb{R},\mathrm{sp}})=\Omega\left(\frac{1}{N^3}\right)$, where we redefine $N=T_{\mathrm{sp}}+M+1$.
    Here we reindex the padded circuit so that $U_t=I$ for $1\leq t\leq M$ and $U_{M+j}$ is the $j$-th gate of the sparse circuit.
    Taking the guiding state as 
    \begin{equation*}
        \ket{c_\text{hist}}
        =\ket{0^{n_{\mathrm{sp}}}}\otimes
        \frac{1}{\sqrt{M+1}} \sum_{t=0}^{M} \ket{t},
    \end{equation*}
    we obtain
    \begin{equation*}
        |\braket{\psi_\text{hist}|c_\text{hist}}|^2
        =\frac{M+1}{N} =1-\frac{T_{\mathrm{sp}}}{N} \geq 1-\frac{1}{\mathrm{poly}(n)}.
    \end{equation*}

    As in the previous section, we introduce one additional ancilla qubit, add the perturbation term $V(\theta)$ in Eq.~\eqref{eq:perturbation}, and obtain
    \begin{equation*}
        H(\theta) 
        = H_\text{hist}^{\mathbb{R},\mathrm{sp}}
        +r V(\theta)
    \end{equation*}
    However, $V(\theta)$ contains imaginary number elements, so we apply the fixed single-qubit rotation $R_a=SH$.
    Since
    \begin{equation*}
        R_aX_aR_a^\dagger=Z_a,
        \quad
        R_aY_aR_a^\dagger=X_a,
    \end{equation*}
    the Hamiltonian is transformed to
    \begin{align}\label{eq:real-5-local}
        H^{\mathbb{R},\mathrm{sp}}(\theta)
        =
        R_a H(\theta) R_a^\dagger
        =
        H_\text{hist}^{\mathbb{R},\mathrm{sp}}
        +
        r V^{\mathbb{R},\mathrm{sp}}(\theta)
        =
        H_\text{hist}^{\mathbb{R},\mathrm{sp}}
        +r 
        \left[
            \left(
                \frac{3}{2}-\Pi_1+\cos{\theta}
            \right)\otimes Z_a
            +
            \sin{\theta} X_a
        \right]
        \otimes \ket{1}\bra{1}_{C_{T+M}}.
    \end{align}
    This transformation is independent of $\theta$ and therefore leaves
    the Berry phase unchanged.
    $H^{\mathbb{R},\mathrm{sp}}(\theta)$ has a symmetry of $Z_aH^{\mathbb{R},\mathrm{sp}}(\theta)Z_a^\dagger=H^{\mathbb{R},\mathrm{sp}}(-\theta)$.
    Note that adding this perturbation term does not affect the spatial sparsity of the interaction graph since the output qubit, the final clock qubit, and the ancilla qubit for the perturbation term can be placed near the end of the computation.

    Then we can apply the same analysis of the spectral gap, Berry phase and guiding state to $H^{\mathbb{R},\mathrm{sp}}(\theta)$ as in the proof of Theorem~\ref{thm:hardness-BPE}.
    The spectral gap is lower bounded by $\Delta_{\min} = \frac{2r}{3N} - O(N^3r^2) \geq 1/\mathrm{poly}(n)$ for all $\theta$ as far as $r = o(1/N^4)$.
    The Berry phase satisfies that
    \begin{equation}\label{eq:berry-quantized-physical}
        \theta_B = 
        \begin{cases}
            \pi & \text{if } x\in L_\text{yes} \\
            0 & \text{if } x\in L_\text{no}.
        \end{cases}
    \end{equation}
    And the guiding state $\ket{c}=\ket{c_\text{hist}}\otimes \ket{1}_a$ satisfies 
    \begin{equation*}
        |\braket{\psi_0^{\mathbb{R},\mathrm{sp}}(0)|c}|^2
        \geq 1-\frac{1}{\mathrm{poly}(n)},
    \end{equation*}
    where $\ket{\psi_0^{\mathbb{R},\mathrm{sp}}(0)}$ is the unique ground state of $H^{\mathbb{R},\mathrm{sp}}(0)$.

    \proofpara{Encoding to the global symmetric system}
    Now we encode $H^{\mathbb{R},\mathrm{sp}}(\theta)$ to the system where the symmetry operator is a reflection of the entire system.
    This property makes it easy to ensure that the Hamiltonian preserves the symmetry under the reduction to physical Hamiltonians.

    We place two copies of $H_\text{hist}^{\mathbb{R},\mathrm{sp}}$ on two spatially separated registers, denoted by $L$ and $R$. 
    We write the corresponding Hamiltonians as $H_{\text{hist},L}^{\mathbb{R},\mathrm{sp}}$ and $H_{\text{hist},R}^{\mathbb{R},\mathrm{sp}}$, and denote the one-qubit projectors for the final clock and the output qubit by $\ket{1}\bra{1}_{C_{T+M},L}$, $\ket{1}\bra{1}_{C_{T+M},R}$ and $\Pi_{1,L}$, $\Pi_{1,R}$, respectively.
    We place two additional qubits, $b_1$ and $b_2$, near the output qubits and final clock qubits of the $L$ and $R$ registers, respectively, and define 
    \begin{equation*}
        H_\mathcal{C} = 
        \frac{I+Z_{b_1}Z_{b_2}}{2}.
    \end{equation*}
    We regard its two-dimensional ground space $\mathcal{C} = \mathrm{span}\{\ket{01},\ket{10}\}$ as a logical subspace and encode $V^{\mathbb{R},\mathrm{sp}}(\theta)$ in this subspace.
    Specifically, we choose the logical basis states as
    \begin{equation*}
        \ket{0_{\mathcal{C}}} \coloneqq
        \frac{\ket{01}+\ket{10}}{\sqrt{2}},
        \quad 
        \ket{1_{\mathcal{C}}} \coloneqq
        \frac{\ket{01}-\ket{10}}{\sqrt{2}},
    \end{equation*}
    and define the logical operators as
    \begin{equation*}
        Z_\mathcal{C} \coloneqq
        \frac{I+X_{b_1}X_{b_2}+Y_{b_1}Y_{b_2}+Z_{b_1}Z_{b_2}}{2}
        \left( =\mathrm{SWAP}_{b_1,b_2} \right), 
        \quad 
        X_\mathcal{C} \coloneqq
        \frac{Z_{b_1}-Z_{b_2}}{2}.
    \end{equation*}
    Then we define the encoded perturbation term
    \begin{equation*}
        V_{\mathcal{C}}^{\mathbb{R},\mathrm{sp}}(\theta)
        \coloneqq
        \frac{1}{2}\sum_{\alpha=L,R}
        \left[  
        \left(
        \frac{3}{2}-\Pi_{1,\alpha} +\cos{\theta}
        \right)
        \otimes Z_\mathcal{C}
        +\sin{\theta} X_\mathcal{C}
        \right]
        \otimes \ket{1}\bra{1}_{C_{T+M},\alpha},
    \end{equation*}
    and the parametrized Hamiltonian
    \begin{equation*}
        H_{\mathcal{C}}^{\mathbb{R},\mathrm{sp}}(\theta)
        \coloneqq
        H_{\text{hist},L}^{\mathbb{R},\mathrm{sp}} + H_{\text{hist},R}^{\mathbb{R},\mathrm{sp}}
        +H_\mathcal{C}
        + r V_{\mathcal{C}}^{\mathbb{R},\mathrm{sp}}(\theta).
    \end{equation*}

    The two-dimensional ground space of $H_{\text{hist},L}^{\mathbb{R},\mathrm{sp}} 
    +H_{\text{hist},R}^{\mathbb{R},\mathrm{sp}} + H_\mathcal{C}$
    is spanned by 
    $\{\ket{\psi_{\text{hist}}}_L \ket{\psi_{\text{hist}}}_R\ket{01}_{b_1,b_2}, \ket{\psi_{\text{hist}}}_L \ket{\psi_{\text{hist}}}_R\ket{10}_{b_1,b_2}\}$ and its spectral gap is $\Delta(H_{\text{hist},L}^{\mathbb{R},\mathrm{sp}} 
    +H_{\text{hist},R}^{\mathbb{R},\mathrm{sp}} + H_\mathcal{C})=\Omega(1/N^3)$.
    Let $P_0$ be a projector onto this ground space.
    Then the first order effective Hamiltonian is
    \begin{equation*}
        P_0 V_{\mathcal{C}}^{\mathbb{R},\mathrm{sp}}(\theta) P_0
        = \ket{\psi_{\text{hist}}}\bra{\psi_{\text{hist}}}_L \otimes \ket{\psi_{\text{hist}}}\bra{\psi_{\text{hist}}}_R \otimes 
        \frac{1}{N} \left[  
        \left(
        \frac{3}{2}-p_1 +\cos{\theta}
        \right)
        \otimes Z_\mathcal{C}
        +\sin{\theta} X_\mathcal{C}
        \right].
    \end{equation*}
    This is essentially the same form as Eq.~\eqref{eq:first-order-SSH}, and therefore we can apply the same analysis to $H_{\mathcal{C}}^{\mathbb{R},\mathrm{sp}}(\theta)$.
    The spectral gap is lower bounded by $\Delta_{\min} = \frac{2r}{3N} - O(N^3r^2) \geq 1/\mathrm{poly}(n)$ for all $\theta$ as far as $r = o(1/N^4)$.
    The guiding state $\ket{c}=\ket{c_\text{hist}}_L \ket{c_\text{hist}}_R \otimes \ket{1_\mathcal{C}}$ satisfies 
    \begin{equation*}
        |\braket{\psi_{0,\mathcal{C}}^{\mathbb{R},\mathrm{sp}}(0)|c}|^2
        \geq 1-\frac{1}{\mathrm{poly}(n)},
    \end{equation*}
    where $\ket{\psi_{0,\mathcal{C}}^{\mathbb{R},\mathrm{sp}}(0)}$ is the unique ground state of $H_{\mathcal{C}}^{\mathbb{R},\mathrm{sp}}(0)$.

    Let $R_{\mathrm{hist}}$ denote the operator that swaps the entire $L$ and $R$ registers.
    We obtain 
    \begin{equation*}
        (R_{\mathrm{hist}}\mathrm{SWAP}_{b_1,b_2})H_{\mathcal{C}}^{\mathbb{R},\mathrm{sp}}(\theta)(R_{\mathrm{hist}}\mathrm{SWAP}_{b_1,b_2})^\dagger = H_{\mathcal{C}}^{\mathbb{R},\mathrm{sp}}(-\theta),
    \end{equation*}
    and therefore the Berry phase satisfies Eq.~\eqref{eq:berry-quantized-physical}.
    We note that the operator $R_{\mathrm{hist}}\mathrm{SWAP}_{b_1,b_2}$ simply swaps the qubits in the left half with the corresponding qubits in the right half.
    Therefore this symmetry would be exactly preserved as long as the perturbative gadgets are applied in exactly the same manner on the left and right halves, even though its simulation of the target Hamiltonian is approximate.

    \proofpara{Reduction to the spatially sparse 2-local Pauli interaction without $Y$ terms}
    We follow Lemma~40 and Theorem~41 of Ref.~\cite{cubitt2018universal} and Ref.~\cite{oliveira2005complexity}.
    As noted in the previous paragraph, the symmetry is preserved by assigning identical gadgets (ancilla qubits and Hamiltonian terms) to each pair of corresponding terms exchanged under the swap of the $L$ and $R$ registers.
    At reduction steps that introduce mediator qubits, terms inherited from the \(L\) and \(R\) registers are kept separately labelled even when they coincide after the Pauli expansion, so that the corresponding gadgets can still be chosen as a pair exchanged by the SWAP. 
    An interaction fixed by the SWAP is similarly split into two half-strength copies exchanged by the SWAP.

    We now reduce the Hamiltonian $H_{\mathcal{C}}^{\mathbb{R},\mathrm{sp}}(\theta)$ to a spatially sparse 2-local Pauli Hamiltonian with no $Y$ terms.
    The real 6-local Hamiltonian $H_{\mathcal{C}}^{\mathbb{R},\mathrm{sp}}(\theta)$ for each $\theta$ can first be simulated by a 7-local Hamiltonian whose Pauli decomposition contains no Pauli $Y$ operators (Lemma~40 of Ref.~\cite{cubitt2018universal}).
    We then apply subdivision gadgets followed by 3-to-2 gadgets~\cite{oliveira2005complexity} to obtain a spatially sparse 2-local Hamiltonian of the form $H_{\mathcal{S}_{XZ}}(\theta) = \sum_{i<j}\alpha_{ij}(\theta)A_{ij} + \sum_k (\beta_k(\theta)X_k + \gamma_k(\theta) Z_k)$, where $A_{ij}\in \{XX,XZ,ZX,ZZ\}$.
    We denote the interaction set by $S_{XZ}=\{XX,XZ,ZX,ZZ,X,Z\}$.
    Applying each gadget together with its copy on the term exchanged by the current SWAP permutation, we obtain a new operator $R_{\mathcal{S}_{XZ}}$ that swaps $L$ and $R$ registers, and it satisfies
    \begin{equation*}
        R_{\mathcal{S}_{XZ}} H_{\mathcal{S}_{XZ}}(\theta) R_{\mathcal{S}_{XZ}}^\dagger = H_{\mathcal{S}_{XZ}}(-\theta),
    \end{equation*}
    exactly after this step.

    At these stages, signed coefficients can be placed in one term of each gadget, giving polynomial dependence on those coefficients without taking their roots.
    For example, consider a parameter dependent term $h = j(\theta) P \otimes Q$ with Pauli strings $P$ and $Q$.
    Then, adding a mediator qubit $m$, the subdivision gadget $\Delta H_0 + \sqrt{\Delta} H_2 + H_1$, with $H_0 = \ket{1}\bra{1}_m$, $H_2= (P\otimes X_m - j(\theta) Q\otimes X_m)/\sqrt{2}$, and $H_1 = (1+j(\theta)^2)I/2$, simulates the target Pauli string
    \begin{equation*}
        \Pi_- H_1 \Pi_- - \Pi_- H_2 \Pi_+ H_0^{-1} \Pi_+ H_2 \Pi_- = 
        \frac{1+j(\theta)^2}{2} I - \frac{1}{2}(P-j(\theta) Q)^2 = j(\theta) P\otimes Q =h,
    \end{equation*}
    where $\Pi_-$ is a projector onto the ground space of $H_0$ and $\Pi_+=I-\Pi_-$.
    Therefore this gadget simulates $h$ (see, e.g., Lemma 37 of Ref.~\cite{cubitt2018universal}).
    The coefficients can also be encoded in the same way for both the 3-to-2 gadget and the gadget used to eliminate $Y$ terms.
    Choosing sufficiently large $\Delta$ independently of $\theta$, then Lemmas 37–39 of Ref.~\cite{cubitt2018universal} yield simulation guarantees that hold uniformly over $\theta$.

    \proofpara{Reduction to the spatially sparse $\{XX+YY+ZZ\}$ or $\{XX+YY\}$-Hamiltonian}
    Next, we simulate the $\mathcal{S}_{XZ}$ interaction with Heisenberg ($\{XX+YY+ZZ\}$) interactions or $XY$ ($\{XX+YY\}$) interactions.
    We mainly follow Theorem 42 of Ref.~\cite{cubitt2018universal}, but we perform a preprocessing before the reduction in order to maintain smoothness of the parameter.
    This is because the encoding of Theorem 42 of Ref.~\cite{cubitt2018universal} uses square roots of interaction strength, which may lead to the divergence of derivative of the Hamiltonian.

    Consider a single two-local Pauli term $j(\theta)P_uQ_v$ for $P,Q\in\{X,Z\}$.
    Choose an upper-bound $M$ such that $\sup_{\theta}|j(\theta)|+1\leq M$, and we take 
    \begin{equation*}
        C_{\pm}(\theta) \coloneqq \frac{M\pm j(\theta)}{2}.
    \end{equation*}
    Then we have $C_{+}(\theta) -C_{-}(\theta) =j(\theta)$ and $C_{\pm}(\theta)\geq 1/2$ for all $\theta$.
    For this Pauli term, we introduce two mediator qubits $m_{+}$ and $m_{-}$, and define $H_{0,m} = \ket{1}\bra{1}_{m_{+}}
    + \ket{1}\bra{1}_{m_{-}}$, and
    \begin{equation*}
        B_m(\theta)
        =
        \sqrt{\frac{C_{+}(\theta)}{2}}
        (P_u-Q_v)X_{m_{+}}
        +
        \sqrt{\frac{C_{-}(\theta)}{2}}
        (P_u+Q_v)X_{m_{-}}.
    \end{equation*}
    We then implement the gadget Hamiltonian
    \begin{equation*}
        \Delta H_{0,m} + \sqrt{\Delta}B_m(\theta) + M I.
    \end{equation*}
    Let $P_{0,m}$ be a projector onto the ground space of $H_{0,m}$.
    By the direct calculation, the first order effective Hamiltonian is $P_{0,m} (M I) P_{0,m} = M P_{0,m}$ and the second order effective Hamiltonian is 
    \begin{align*}
        - P_{0,m}B_m(\theta)(I-P_{0,m})
        H_{0,m}^{-1}
        (I-P_{0,m})B_m(\theta)P_{0,m}
        &= \left[ -\frac{C_{+}(\theta)}{2}(P_u-Q_v)^2 
        - \frac{C_{-}(\theta)}{2}(P_u+Q_v)^2\right] P_{0,m}\\
        &= j(\theta)P_u Q_v P_{0,m} - M P_{0,m}.
    \end{align*}
    Therefore $\Delta H_{0,m} + \sqrt{\Delta}B_m(\theta) + M I$ simulates $j(\theta)P_u Q_v$.
    In this way, since $\sqrt{C_\pm(\theta)/2}\geq 1/2$ for all $\theta$, the coefficient is smooth for all $\theta$.
    This property is essential for the reduction to Heisenberg and XY interaction.
    Applying this gadget to the terms with coefficients depending on $\theta$, and introducing the mediator qubits and gadget Hamiltonians in exactly the same manner for each pair of Pauli terms exchanged by the SWAP operation, preserves the symmetry under the SWAP operation.

    Now we can complete the reduction to the $\mathcal{S}_0=\{XX+YY+ZZ\}$ or $\{XX+YY\}$-Hamiltonians simply by applying Theorem 42 of Ref.~\cite{cubitt2018universal}.
    This reduction uses a subspace encoding in which each logical qubit is encoded into the ground space of a fixed four-qubit Heisenberg or XY Hamiltonian, and the desired two-qubit logical interactions are generated perturbatively by interactions between the corresponding four-qubit blocks.
    As in Eq. (199) of Ref.~\cite{cubitt2018universal}, the second-order perturbation term takes the form $\propto \Pi_{-}^\text{tot}H_2^2\Pi_{-}^\text{tot}$.
    Therefore, to encode a coefficient $j(\theta)$, one would naively choose as $H_2 = \sqrt{|j(\theta)|}H_2'$ using the gadget for the corresponding sign.
    However, this choice may lead to a divergent derivative when \(j(\theta)\) approaches zero.
    The gadget introduced above avoids this problem.
    Corresponding logical qubits in the \(L\) and \(R\) registers are encoded using identical four-qubit blocks. 
    The perturbative gadgets associated with terms exchanged by the SWAP are chosen to be exchanged in the same way while a interaction fixed by the SWAP is implemented using a SWAP-invariant perturbative interaction. 
    Hence the SWAP symmetry is preserved throughout the subspace encoding.
    Denoting the resulting simulator Hamiltonian by $H_{\mathcal{S}_0}(\theta)$ and the corresponding SWAP operator by $R_{\mathcal{S}_0}$, we obtain
    \begin{equation*}
        R_{\mathcal{S}_0}H_{\mathcal{S}_0}(\theta)R_{\mathcal{S}_0}^\dagger
        = H_{\mathcal{S}_0}(-\theta).
    \end{equation*}

    \proofpara{Reduction to the $\mathcal{S}_0$-Hamiltonians on a 2D square lattice}
    Following Lemma 47 of Ref.~\cite{cubitt2018universal}, the reduction to a Hamiltonian on a 2D square lattice is achieved by applying subdivision, fork, and crossing gadgets of Refs.~\cite{oliveira2005complexity, piddock2015complexity}.
    The symmetry can also be preserved by applying identical gadgets to pairs of terms exchanged by the SWAP operator.
    Denoting the Hamiltonian obtained after the reduction by $H_{\mathrm{2D}}(\theta)$ and the corresponding SWAP operator that exchanges the $L$ and $R$ registers by $R_{\mathrm{2D}}$, we have
    \begin{equation*}
        R_{\mathrm{2D}} H_{\mathrm{2D}}(\theta) R_{\mathrm{2D}}^\dagger 
        = H_{\mathrm{2D}}(-\theta).
    \end{equation*}

    \proofpara{Reduction to arbitrary non-2SLD $\mathcal{S}$-Hamiltonian on a 2D square lattice}
    We apply Theorem 43 of Ref.~\cite{cubitt2018universal}, including its square-lattice implementation from Ref.~\cite{piddock2015complexity}. 
    Since all gadgets in the proof are local mediator or constant-size subspace-encoding gadgets, we apply identical gadgets to every pair of target interactions exchanged by $R_{\mathrm{2D}}$.
    Hence the SWAP symmetry is preserved.
    For the corresponding SWAP $R_{\mathcal{S}}$ and the reduced Hamiltonian $H_\mathcal{S}(\theta)$, we obtain
    \begin{equation}\label{eq:final-symmetry}
        R_{\mathcal{S}} H_\mathcal{S}(\theta) R_{\mathcal{S}}^\dagger = H_\mathcal{S}(-\theta).
    \end{equation}
    Recall that the preprocessing above ensures that every parameter-dependent coefficient to which a square root is applied is uniformly bounded away from zero.
    In the subsequent reductions, square roots are taken only of coefficient with this property, possibly multiplied by \(\theta\)-independent factors.
    Therefore, the reduced Hamiltonian $H_\mathcal{S}$ is still smooth for all $\theta$.

    \proofpara{Simulation errors, Berry phase, and normalization}
    Let $V_0$ be the composition of the fixed encoding used above, which consists of adding the ground states of mediator penalties, encoding each qubit in the ground space of Heisenberg or XY four-qubit gadgets, and applying the fixed encodings for $\mathcal S$.
    As shown in Ref.~\cite{bravyi2017complexity} (and reviewed in Ref.~\cite{cubitt2018universal}), errors of effective-Hamiltonian $\varepsilon_{\mathrm{sim}}$ and encoding errors $\eta_{\mathrm{sim}}$ can be inverse polynomially suppressed with polynomial large penalty strength.
    Let $g$ be the minimum gap of $H_{\mathcal C}^{\mathbb R,\mathrm{sp}}(\theta)$.
    With sufficiently small $\varepsilon_{\mathrm{sim}}<g/4$, the final simulator Hamiltonian has a unique ground state and a spectral gap at least $g/2$.
    By Lemma 2 and Lemma 3 of Ref.~\cite{bravyi2017complexity}, the ground state of the simulator Hamiltonian $\ket{\tilde{\psi}_0(\theta)}$ and the target ground state $\ket{\psi_0(\theta)}$ satisfy
    \begin{equation}\label{eq:simulation-error}
        \sup_{\theta} \left\| \ket{\tilde{\psi}_0(\theta)} - V_0 \ket{\psi_0(\theta)} \right\| \leq \eta_{\mathrm{sim}} + O(\varepsilon_{\mathrm{sim}}/g).
    \end{equation}

    The symmetric gadget construction preserves the symmetry in Eq.~\eqref{eq:final-symmetry} and gives
    \begin{equation}\label{eq:symmetry-simulation}
        V_0^\dagger R_\mathcal{S}V_0 =  R_{\mathrm{hist}}\mathrm{SWAP}_{b_1,b_2}  .
    \end{equation}
    By symmetry, $\ket{\tilde{\psi}_0(0)}$, $\ket{\tilde{\psi}_0(\pi)}$, $V_0 \ket{\psi_0(0)}$ and $V_0 \ket{\psi_0(\pi)}$ are eigenstates of $R_\mathcal{S}$.
    If $\ket{\tilde{\psi}_0(\theta)}$ and $V_0 \ket{\psi_0(\theta)}$ have opposite eigenvalues at $\theta=0$ or $\pi$, then the distance between these two states in Eq.~\eqref{eq:simulation-error} becomes $\sqrt{2}$.
    Therefore, the eigenvalues of $R_\mathcal{S}$ with respect to the ground state at $\theta=0$ and $\pi$ are preserved as long as the simulation error in Eq.~\eqref{eq:simulation-error} is smaller than $\sqrt{2}$.
    Hence the Berry phase is also preserved as Eq.~\eqref{eq:berry-quantized-physical}.

    The choices of coefficients above give polynomial upper-bounds on all coefficients and their first and second derivatives. 
    Divide the final Hamiltonian by the maximum among the upper-bound on the norm of local terms and the sums of norms of the first and second derivatives. 
    This makes the norm of each local term and $B_1,B_2$ at most one, preserves the eigenvectors, symmetry, and Berry phase, and leaves an inverse-polynomial gap.

    \proofpara{Guiding state}
    The result about the guiding state is a direct consequence of Ref.~\cite{cade2022improved}.
    Setting the guiding state as $\ket{\tilde{c}}\coloneqq V_0 \ket{c}$, we have 
    \begin{equation*}
        |\braket{\tilde{\psi}_0(0)|\tilde{c}}|^2 \geq 1 - 1/\mathrm{poly}(n),
    \end{equation*}
    as long as the simulation error in Eq.~\eqref{eq:simulation-error} is inverse polynomially small.
    By construction, $\ket{\tilde{c}}$ is a semi-classical encoded state and efficient sampling-and-query-access to this state can be implemented classically as shown in Ref.~\cite{cade2022improved}.
\end{proof}

\section{Classical algorithms for Berry phase estimation}\label{sec:dequantize}
In this section, we provide an efficient classical algorithm for Berry phase estimation in the constant spectral gap regime and analyze its explicit complexity.
To show this, we first approximate the quasi-adiabatic continuation operator by a spatial truncation, and then provide a general classical algorithm for simulating time-dependent Hamiltonians using truncated and discretized Dyson series and the dequantized QSVT for sparse-matrices~\cite{gharibian2022dequantizing}.
Combining these two ingredients, we obtain an explicit classical algorithm for estimating Berry phase.
In particular, when the Berry phase is promised to be $0$ or $\pi$, the classical algorithm can estimate the Berry phase in polynomial time provided that $B_1=O(1)$ and $\gamma \geq 1/2 + \Omega(1)$ and in quasi-polynomial time provided that $B_1=O(1)$ and $\gamma \geq 1/2$.

\subsection{Quasi-adiabatic continuation and spatial truncation}\label{subsec:qac}
Consider the qubit system on a fixed $d$-dimensional lattice.
Write
\begin{equation}\label{eq:local-decomposition}
 H(\theta)=\sum_{Z\in\mathcal E}h_Z(\theta),
 \quad
 \sup_{\theta,Z}\|h_Z(\theta)\|\leq 1,
 \quad
 m\coloneqq |\mathcal E|=O(n),
 \quad
 \sup_{\theta,x}\sum_{Z\ni x}\|h_Z(\theta)\|=O(1).
\end{equation}
Note that the norm of the whole Hamiltonian can be $\|H(\theta)\|\leq m =O(n)$.
We assume that a ground state is non-degenerate along the loop and its spectral gap is lower-bounded by $\Delta_{\min}$. 
We denote the upper-bound of first and second derivatives of the Hamiltonian as $B_1$ and $B_2$.
For any time-dependent Hamiltonian $G(t)$, we define its time evolution operator as
\begin{equation*}
    U_G(T,0)=\mathcal{T} \exp{\left( -i \int_0^T G(t) \mathrm{d}t \right)}.
\end{equation*}

The quasi-adiabatic continuation operator~\cite{hastings2005quasiadiabatic, hastings2010quasi}
\begin{equation*}
    D(\theta) = \int_{-\infty}^{\infty} 
    W_{\Delta_\text{min}}(t) e^{itH(\theta)} \dot{H}(\theta) e^{-itH(\theta)} \mathrm{d}t,
\end{equation*}
would yield the Berry phase
\begin{equation*}
    U_{D}(2\pi,0)\ket{\psi_0(0)}
    = e^{i\theta_B} \ket{\psi_0 (0)}, 
\end{equation*}
when the filter function is real and odd and satisfies
\begin{equation}\label{eq:filter-condition}
 \widehat W_{\Delta_{\min}}(0)=0,\qquad
 \widehat W_{\Delta_{\min}}(\omega)=-\frac{i}{\omega}\quad(|\omega|\geq\Delta_{\min}),
\end{equation}
where $\widehat W_{\Delta_{\min}}(\omega)=\int_{\mathbb R}W_{\Delta_{\min}}(t)e^{i\omega t}\mathrm{d}t$.
Ref.~\cite{hastings2010quasi} constructs an explicit function that satisfies these properties.
As a direct consequence of Corollary 7 in Ref.~\cite{hastings2010quasi}, we have
\begin{equation}\label{eq:filter-moments}
 |W_{\Delta_{\min}}(t)|\leq O\left(e^{-\sqrt{\Delta_{\min}|t|}}\right),\qquad
 w_0:=\int|W_{\Delta_{\min}}(t)|\,\mathrm{d}t\leq O(\Delta_{\min}^{-1}),\qquad
 w_1:=\int|tW_{\Delta_{\min}}(t)|\,\mathrm{d}t\leq O(\Delta_{\min}^{-2}).
\end{equation}
Then, by direct calculation, the norms of $D(\theta)$ and its derivative are bounded as
\begin{equation}\label{eq:qac-norm-derivative}
 \sup_\theta\|D(\theta)\|\leq w_0B_1\leq O\left(\frac{B_1}{\Delta_{\min}}\right),\qquad
 \sup_\theta\|\dot D(\theta)\|
 \leq w_0B_2+2w_1B_1^2
 \leq O\left(\frac{B_2}{\Delta_{\min}}+\frac{B_1^2}{\Delta_{\min}^2}\right).
\end{equation}

Now we expand the local terms in Pauli strings
\begin{equation*}
    H(\theta)=a_I(\theta)I+\sum_{P\in\mathcal P}a_P(\theta)P.
\end{equation*}
We have separated the identity term from the traceless terms, and each qubit participates in $O(1)$ Pauli terms.
As shown in Lemma 5 of Ref.~\cite{harrow2017extremal}, we have
\begin{equation*}
    \sum_{P\in\mathcal P}| a_P(\theta)|\leq O(\| H(\theta)\|).
\end{equation*}
Applying the same bound to the traceless parts of $\dot{H}$ and $\ddot{H}$, we obtain
\begin{equation}\label{eq:variation-pauli-bound}
    \sum_{P\in\mathcal P}|\dot a_P(\theta)|\leq O(\|\dot H(\theta)\|) =O(B_1),
    \quad \text{and} \quad
    \sum_{P\in\mathcal P}|\ddot a_P(\theta)|\leq O(\|\ddot H(\theta)\|) =O(B_2).
\end{equation}

For each Pauli string $P$, let $Y_{P,R}$ be its radius-$R$ neighborhood and let
$H_{P,R}(\theta)
=
\sum_{\substack{Q\in\mathcal{P}: \operatorname{supp}(Q)\subseteq Y_{P,R}}}
a_Q(\theta) Q$. 
Define
\begin{equation}\label{eq:qac-truncation}
 D_{P,R}(\theta)= \dot{a}_P(\theta) \int_{\mathbb R}W_{\Delta_{\min}}(t)
 e^{itH_{P,R}(\theta)} P e^{-itH_{P,R}(\theta)}\mathrm{d}t,
 \qquad D_R(\theta)=\sum_{P\in\mathcal P}D_{P,R}(\theta).
\end{equation}
Note that the term proportional to the identity operator in $H(\theta)$ does not contribute to $D_R(\theta)$ since the filter function is chosen such that $\int_{\mathbb R}W_{\Delta_{\min}}(t) \mathrm{d}t = \widehat W_{\Delta_{\min}}(0)=0$.
Therefore we can assume that $H_{P,R}(\theta)$ does not contain the identity term.
The following lemma bounds the error from this truncation.
\begin{lem}\label{lem:qac-local}
It holds that
\begin{equation}\label{eq:qac-local-error}
 \sup_\theta\|D(\theta)-D_R(\theta)\|
 \leq O\left(\frac{B_1}{\Delta_{\min}}e^{-\Omega\left(\sqrt{\Delta_{\min} R}\right)}\right).
\end{equation}
Thus an error at most $\zeta>0$ is obtained with
\begin{equation}\label{eq:radius-zeta}
 R=O\left(\frac {1}{\Delta_{\min}}
      \log^2\left(\frac{B_1}{\Delta_{\min} \zeta} \right)\right).
\end{equation}
The resulting $D_R$ has at most $m2^{O(R^d)}$ nonzero entries per row and column.
Also it holds that
\begin{equation}
    \sup_\theta \|D_R(\theta)\|\leq O\left(\frac{B_1}{\Delta_{\min}}\right),
    \quad \text{and} \quad 
    \sup_\theta\|\dot D_R(\theta)\| =
    O\left(\frac{B_2}{\Delta_{\min}}+
             \frac{B_1^2}{\Delta_{\min}^2}\right).
\end{equation}
\end{lem}
\begin{proof}
    The following bound is a simplification of the analysis in Ref.~\cite{hastings2010quasi}.
    \begin{align*}
        \sup_\theta\|D(\theta)-D_R(\theta)\| 
        &\leq
        \sup_\theta \sum_{P\in\mathcal P} |\dot{a}_P(\theta)| \int_{\mathbb R} |W_{\Delta_{\min}}(t)| 
        \left\| e^{itH(\theta)} P e^{-itH(\theta)} - e^{itH_{P,R}(\theta)} P e^{-itH_{P,R}(\theta)} \right\| \mathrm{d}t \\
        &\leq \sup_\theta \sum_{P\in\mathcal P} |\dot{a}_P(\theta)| \int_{|t|\leq R/2v_{\mathrm{LR}}} |W_{\Delta_{\min}}(t)| 
        \left\| e^{itH(\theta)} P e^{-itH(\theta)} - e^{itH_{P,R}(\theta)} P e^{-itH_{P,R}(\theta)} \right\| \mathrm{d}t\\
        &\qquad +\sup_\theta \sum_{P\in\mathcal P} |\dot{a}_P(\theta)| \int_{|t|> R/2v_{\mathrm{LR}}} |W_{\Delta_{\min}}(t)| 
        \left\| e^{itH(\theta)} P e^{-itH(\theta)} - e^{itH_{P,R}(\theta)} P e^{-itH_{P,R}(\theta)} \right\| \mathrm{d}t \\
        & \leq \sup_\theta \sum_{P\in\mathcal P} |\dot{a}_P(\theta)| \int_{|t|\leq R/2v_{\mathrm{LR}}} |W_{\Delta_{\min}}(t)| 
        \mathrm{d}t O(e^{-\mu R}) 
        + 2\sup_\theta \sum_{P\in\mathcal P} |\dot{a}_P(\theta)| \int_{|t|> R/2v_{\mathrm{LR}}} |W_{\Delta_{\min}}(t)| 
         \mathrm{d}t \\
        &\leq O\left(\frac{B_1}{\Delta_{\min}} e^{-\mu R}\right) + O(B_1) \int_{|t|> R/2v_{\mathrm{LR}}} |W_{\Delta_{\min}}(t)| 
         \mathrm{d}t,
    \end{align*}
    where $v_{\mathrm{LR}}$ is Lieb-Robinson velocity and $\mu$ is a constant.
    Recalling that $|W_{\Delta_{\min}}(t)|\leq O\left(e^{-\sqrt{\Delta_{\min}|t|}}\right)$, we obtain
    \begin{align*}
        \int_{|t|> R/2v_{\mathrm{LR}}} |W_{\Delta_{\min}}(t)| 
         \mathrm{d}t
        \leq O\left(\frac{\sqrt{\Delta_{\min}R/v_{\mathrm{LR}}}}{\Delta_{\min}} e^{-\Omega(\sqrt{\Delta_{\min}R/v_{\mathrm{LR}}})}\right)
        = O\left(\frac{1}{\Delta_{\min}} e^{-\Omega(\sqrt{\Delta_{\min}R/v_{\mathrm{LR}}})}\right),
    \end{align*}
    where the last equality holds since the polynomial prefactor can be absorbed into the exponential function.
    Therefore we obtain
    \begin{equation*}
        \sup_\theta\|D(\theta)-D_R(\theta)\| 
        \leq O\left(\frac{B_1}{\Delta_{\min}} e^{-\Omega(\sqrt{\Delta_{\min}R})}\right).
    \end{equation*}
    
    Now we calculate the norms of $D_R(\theta)$ and $\dot{D}_R(\theta)$.
    Using Eq.~\eqref{eq:filter-moments} and Eq.~\eqref{eq:variation-pauli-bound}, we have
    \begin{equation*}
        \sup_\theta \|D_R(\theta)\|
        \leq \sup_\theta \sum_{P\in\mathcal P}|\dot a_P(\theta)| w_0 \leq O\left(\frac{B_1}{\Delta_{\min}}\right).
    \end{equation*}

    Observe that $\|\dot{H}_{P,R}(\theta)\| \leq \sum_{P\in\mathcal P}|\dot a_P(\theta)| \leq O(B_1)$ since $H_{P,R}(\theta)$ does not have identity term.
    Then, using Eq.~\eqref{eq:variation-pauli-bound} and Eq.~\eqref{eq:filter-moments}, we obtain
    \begin{align*}
        \sup_\theta \|\dot{D}_R(\theta)\|
        &\leq \sup_\theta 
        \left( \sum_{P\in\mathcal P}|\ddot a_P(\theta)| w_0 
        + 2 \sum_{P\in\mathcal P}|\dot a_P(\theta)| \cdot \|\dot{H}_{P,R}(\theta)\|
        \cdot w_1 
        \right) \\
        &\leq O\left(\frac{B_2}{\Delta_{\min}}+
             \frac{B_1^2}{\Delta_{\min}^2}\right).
    \end{align*}
\end{proof}
To evaluate the truncated terms $D_{P,R}(\theta)$, we exactly diagonalize the local matrix $H_{P,R}=\sum_a E_a(\theta) \ket{E_a(\theta)}\bra{E_a(\theta)}$.
Then we can efficiently calculate
\begin{equation*}
    \braket{E_a(\theta)|D_{P,R}(\theta)|E_b(\theta)}=\dot{a}_P(\theta)\widehat{W}_{\Delta_{\min}}(E_a(\theta)-E_b(\theta)) \braket{E_a(\theta)|P|E_b(\theta)},
\end{equation*}
since $H_{P,R}$ acts on $O(R^d)$ qubits and hence is a constant-size matrix when $\Delta_{\min}$, $B_1$, and $\zeta$ are constant.
We only have to diagonalize the truncated matrices and the eigenvectors of the full many-body Hamiltonian are not required. 
We absorb the cost of this diagonalization and evaluation of the fixed scalar function $\widehat W_{\Delta_{\min}}$ into $\bm q(D_R)$. 
If the coefficients and this scalar filter can be evaluated to $b$ bit-precision in $\operatorname{poly}(n,b)$-time, then 
\begin{equation*}
    \bm q(D_R) = m2^{O(R^d)}\operatorname{poly}(n,b).
\end{equation*}

\subsection{Classical algorithm for simulating time-dependent Hamiltonian dynamics}\label{subsec:sparse-time}
Now we review the inner product estimation by sampling, which is a central technique of dequantization.
\begin{lem}\label{lem:inner-product-estimation}
    Suppose that we are given
    \begin{itemize}
        \item query-access to a vector $w$ such that $\|w\|\leq 1$,
        \item sampling-and-query-access to a vector $v$ such that $\|v\|\leq 1$.
    \end{itemize}
    Then we can estimate $\hat{z}\in\mathbb{C}$ such that
    \begin{equation}
    \left|\hat{z} - v^\dagger w \right| \leq \varepsilon,
    \end{equation}
    classically with a probability at least $1-\delta$ in $O(\log{(1/\delta)}(\bm{q}(w)+\bm{sq}(v))/\varepsilon^2)$-time.
\end{lem}
The proof of this lemma is given, for example, in Proposition~1 of Ref.~\cite{gall2023robust}.
This lemma implies that, if we want to estimate $\braket{v|U|u}$ for some unitary, we have to construct query-access to $U\ket{u}$.

The following lemma provides an explicit complexity of dequantized QSVT for time-dependent Hamiltonian simulation.

\begin{lem}\label{lem:dequantized-QSVT-time-dependent}
    Let $T\in \mathbb{R}_+$ be an evolution time and $\alpha, \beta  \in \mathbb{R}_{+}$.
    Take $0<\varepsilon\leq1$ and $0<\delta<1/2$.
    Suppose that we are given 
    \begin{itemize}
        \item query-access to a quantum state $\ket{u}$,
        \item query-access to an $s$-sparse time-dependent Hamiltonian $G(t)$ for each $t\in [0,T]$ such that $\sup_t{\|G(t)\|}\leq \alpha$ and $\sup_t{\|\dot{G}(t)\|}\leq \beta$,
        \item sampling-and-query-access to a quantum state $\ket{v}$.
    \end{itemize}
    Then there exists a classical algorithm that estimates 
    \begin{equation*}
        \braket{v|U_G(T,0)|u},
    \end{equation*}
    with an additive error $\varepsilon$, a failure probability $\delta$ in 
    \begin{equation*}
        O\left( \frac{\log{(1/\delta)}}{\varepsilon^2} 
        \left( \bm{sq}(\ket{v})+ \bm{q}(\ket{u}) K (Ms)^K + \bm{q}(G) K^2 (Ms)^K \right) \right), 
    \end{equation*}
    -time, where
    \begin{equation*}
        K=O\left(
        \alpha T + \log{\left(\frac{1}{\varepsilon}\right)}
        \right), \quad 
        \text{and} \quad 
        M=O\left(
        \frac{\beta T^2}{\varepsilon}
        \right).
    \end{equation*}
\end{lem}
\begin{proof}
    \proofpara{Discretization and truncation of time evolution operator}
    We first divide the evolution time into $M$ intervals, and let
    $\tau:=\frac{T}{M}$,
    $t_j:=j\tau$, and 
    $j=0,\ldots,M$.
    We define the discretized Hamiltonian as
    $G_M(t):=G(t_{j})$ for $t\in[t_j,t_{j+1})$.
    Then we have
    \begin{equation}\label{eq:time-grid-error}
         \|U_G(T,0)-U_{G_M}(T,0)\|
         \leq\int_0^T\|G(t)-G_M(t)\|\mathrm{d}t
         \leq M \int_0^\tau \beta t\mathrm{d}t
         = \frac{\beta T^2}{2M}.
    \end{equation}
    Choosing $M=O(\beta T^2 /\varepsilon)$, the additive error can be upper-bounded by $\varepsilon/3$.
    
    Let $G_j\coloneqq G(t_{j})$ and define
    \begin{equation}\label{eq:discrete-dyson}
         \widetilde U_{K,M}
         =\sum_{k=0}^{K}\frac{(-i\tau)^k}{k!}
           \sum_{j_1,\ldots,j_k=0}^{M-1}
           \mathcal T[G_{j_k}\cdots G_{j_1}],
    \end{equation}
    where $\mathcal T$ orders the grid indices, with later times on the left.
    This is exactly the degree-$K$ Dyson expansion of the piecewise-constant evolution $U_{G_M}(T,0)$.
    Then the truncation error can be written as
    \begin{equation}\label{eq:dyson-approximation}
     \|\widetilde U_{K,M}-U_{G_M}(T,0)\|\leq
     \sum_{k=K+1}^{\infty}\frac{(\tau)^k}{k!}
       \sum_{j_1,\ldots,j_k=0}^{M-1}
       \|\mathcal T[G_{j_k}\cdots G_{j_1}]\|
       \leq \sum_{k=K+1}^{\infty} \frac{(\tau)^k}{k!} (M\alpha)^k
       \leq \sum_{k=K+1}^{\infty} \frac{(\alpha T)^k}{k!} .
    \end{equation}
    When $K+1\geq 2e \alpha T$, we have $\frac{(\alpha T)^k}{k!}\leq \left( \frac{e \alpha T}{k} \right)^k \leq 2^{-k} $ for all $k\geq K+1$.
    Therefore 
    \begin{equation*}
        \|\widetilde U_{K,M}-U_{G_M}(T,0)\|\leq 2^{-K}.
    \end{equation*}
    Choosing $K=O(\alpha T + \log{(1/\varepsilon)})$, the additive error is upper bounded by $\varepsilon/3$.

    \proofpara{Complexity of the classical algorithm}
    Now we have a linear combination of products of sparse Hamiltonians $\widetilde U_{K,M}$ that approximates $U_G(T,0)$ with an additive error of $2\varepsilon/3$.
    Therefore, an estimation of $\braket{v|\widetilde U_{K,M}|u}$ with an additive error $\varepsilon/3$ implies an estimation of $\braket{v|U_G(T,0)|u}$ with an additive error $\varepsilon$.
    Note that we have taken $0<\varepsilon \leq 1$, so $\left\|\widetilde U_{K,M}\ket{u}\right\| \leq 1+\varepsilon/3 <2$.
    Therefore, we may rescale the vector as $\frac{1}{2}\widetilde U_{K,M}\ket{u}$ and estimate $\braket{v|\widetilde U_{K,M}|u}/2$ with an additive error $\varepsilon/6$.
    Multiplying the estimate by $2$ gives an estimate of $\braket{v|\widetilde U_{K,M}|u}$ with additive error $\varepsilon/3$.

    The number of terms in $\widetilde U_{K,M}$ is at most $(K+1) M^K = O(K M^K)$.
    The sparsity of each term is at most $s^{K}$, and therefore, 
    \begin{equation*}
        \bm{q}\left( \widetilde U_{K,M}\ket{u} \right)
        = O \left( \bm{q}(\ket{u}) K (Ms)^K + \bm{q}(G) K^2 (Ms)^K \right) .
    \end{equation*}
    As a result, by Lemma~\ref{lem:inner-product-estimation}, the classical algorithm can estimate $\braket{v|U_G(T,0)|u}$ with an additive error $\varepsilon$ and a failure probability $\delta$ in
    \begin{equation*}
        O\left( \frac{\log{(1/\delta)}}{\varepsilon^2} 
        \left( \bm{sq}(\ket{v})+ \bm{q}(\ket{u}) K (Ms)^K + \bm{q}(G) K^2 (Ms)^K \right) \right), 
    \end{equation*}
    -time.
\end{proof}

\subsection{Classical algorithm for general Berry phase estimation}
We now construct a classical algorithm for estimating a continuous valued Berry phase and analyze its computational complexity.
Specifically, we apply Lemma~\ref{lem:dequantized-QSVT-time-dependent} to the truncated quasi-adiabatic continuation in Eq.~\eqref{eq:qac-truncation} with total evolution time $2\pi$.
The error in the guiding state introduces a bias in the estimate, while the remaining errors are controlled by the truncation of the quasi-adiabatic generator and by the discretization and truncation of the Dyson series.
For classical access to the truncated quasi-adiabatic evolution, we follow Section~\ref{subsec:qac}.

\begin{thm}\label{thm:classical-BPE}
    Suppose that $|\braket{u|\psi_0(0)}|^2 \geq 1-\eta$, with $\eta\geq0$, $0<\varepsilon\leq1$, and $2\eta+\varepsilon/2\leq1/2$.
    Assume that we are given
    \begin{itemize}
        \item sampling-and-query-access to a quantum state $\ket{u}$,
        \item query-access to an $s$-sparse time-dependent Hamiltonian $D_R(\theta)$ for each $\theta\in [0,2\pi]$.
    \end{itemize}
    Then there exists a classical algorithm that can estimate $\theta_B$ with an additive error $\leq 4\eta + \varepsilon$ and a failure probability $\leq \delta$ in 
    \begin{equation}\label{eq:complexity-classical-BPE}
        O\left( \frac{\log{(1/\delta)}}{\varepsilon^2} 
        \left( \bm{sq}(\ket{u})+ \bm{q}(\ket{u}) K (Ms)^K + \bm{q}(D_R) K^2 (Ms)^K \right) \right),
    \end{equation}
    -time, where
    \begin{align*}
        K&= O\left(
        \frac{B_1}{\Delta_{\min}} + \log{\left(\frac{1}{\varepsilon}\right)}
        \right), \\
        M&=  O\left( \frac{1}{\varepsilon} \left(\frac{B_2}{\Delta_{\min}}+\frac{B_1^2}{\Delta_{\min}^2}\right) \right), \\
        s&= O\left( m2^{O(R^d)} \right), \\
        R&=O\left(\frac {1}{\Delta_{\min}}
          \log^2\left(\frac{B_1}{\Delta_{\min} \varepsilon} \right)\right).
    \end{align*}
    For example, if we suppose that $\bm{sq}(\ket{u})=\bm{q}(\ket{u})=\bm{q}(D_R)=\mathrm{poly}(n)$, $\varepsilon=\Theta(1)$, $\Delta_{\min}=\Theta(1)$, $B_1=O(1)$, and $B_2=\mathrm{poly}(n)$, then the classical algorithm runs in polynomial time.
\end{thm}
\begin{proof}
    Recall that the Berry phase is given by
    \begin{equation*}
        e^{i\theta_B} = \braket{\psi_0(0)|U_D(2\pi,0)|\psi_0(0)}.
    \end{equation*}
    Expanding the guiding state as $\ket{u}=\sqrt{p}\ket{\psi_0(0)}+\sqrt{1-p}\ket{u_\perp}$, we obtain
    \begin{equation*}
        \left| \braket{u|U_D(2\pi,0)|u} - e^{i\theta_B} \right|
        \leq 2(1-p) \leq 2\eta,
    \end{equation*}
    since $\braket{u|U_D(2\pi,0)|u}=pe^{i\theta_B}+(1-p)\braket{u_\perp|U_D(2\pi,0)|u_\perp}$.
    The cross terms vanish because $U_D(2\pi,0)$ has $\ket{\psi_0(0)}$ as an eigenvector.
    As shown in Lemma~\ref{lem:qac-local}, the choice of $R=O\left(\Delta_{\min}^{-1}
      \log^2\left(B_1/(\Delta_{\min} \varepsilon) \right)\right)$ ensures that $\|U_D(2\pi,0) - U_{D_R}(2\pi,0)\| \leq \varepsilon/4$.
    Using the algorithm in Section~\ref{subsec:sparse-time} that estimates $\braket{u|U_{D_R}(2\pi,0)|u}$ with an additive error of $\varepsilon/4$, we can obtain an estimate satisfying $|\hat{z}-e^{i\theta_B}|\leq \varepsilon/2 + 2\eta$. 
    For $\varepsilon/2 + 2\eta \leq 1/2$, we have $|\arg{\hat{z}}-\theta_B|_{2\pi}\leq \varepsilon + 4\eta$.

    The runtime is obtained by applying Lemma~\ref{lem:dequantized-QSVT-time-dependent} with the following parameters:
    \begin{align*}
        K&= O\left(
        \alpha + \log{\left(\frac{1}{\varepsilon}\right)}
        \right)
        = O\left(
        \frac{B_1}{\Delta_{\min}} + \log{\left(\frac{1}{\varepsilon}\right)}
        \right), \\
        M&= O\left( \frac{\beta}{\varepsilon} \right) = O\left( \frac{1}{\varepsilon} \left(\frac{B_2}{\Delta_{\min}}+\frac{B_1^2}{\Delta_{\min}^2}\right) \right), \\
        s&= O\left( m2^{O(R^d)} \right), \\
        R&=O\left(\frac {1}{\Delta_{\min}}
          \log^2\left(\frac{B_1}{\Delta_{\min} \varepsilon} \right)\right).
    \end{align*}
    Then the runtime is given by
    \begin{equation*}
        O\left( \frac{\log{(1/\delta)}}{\varepsilon^2} 
        \left( \bm{sq}(\ket{u})+ \bm{q}(\ket{u}) K (Ms)^K + \bm{q}(D_R) K^2 (Ms)^K \right) \right).
    \end{equation*}
\end{proof}

\subsection{Classical algorithm for estimating the quantized Berry phase}\label{subsec:classical-quantized}
For a phase promised to be $0$ or $\pi$, it suffices to determine the sign of the $e^{i\theta_B}$. 
If the ground-state overlap $\gamma$ is strictly larger than $1/2$, its contribution already dominates the real part of the estimated amplitude, that is, $\mathrm{sgn}(e^{i\theta_B})=\mathrm{sgn}(\mathrm{Re}(\braket{u|U_D(2\pi,0)|u}))$.
In the case of $\gamma\geq 1/2$, the following degree-one energy filter suppresses every excited state contribution, which enables us to determine the quantized Berry phase from the sign of the filtered amplitude.

\begin{thm}\label{thm:classical-quantized-BPE}
    Suppose that we are given $|\braket{u|\psi_0(0)}|^2 \geq \gamma$ with $1/2\leq\gamma\leq1$, and the Berry phase $\theta_B\in \{0,\pi\}$. 
    Assume that we are given
    \begin{itemize}
        \item sampling-and-query-access to a quantum state $\ket{u}$,
        \item query-access to an $s$-sparse time-dependent Hamiltonian $D_R(\theta)$ for each $\theta\in [0,2\pi]$.
    \end{itemize}
    Then there exists a classical algorithm that can estimate $\theta_B$ with a failure probability $\leq \delta$ in 
    \begin{equation}
        O\left( \frac{\log{(1/\delta)}}{\varepsilon^{*2}} 
        \left( \bm{sq}(\ket{u})+ \bm{q}(\ket{u}) m2^k K (Ms)^K + \bm{q}(D_R) K^2 (Ms)^K \right) \right),
    \end{equation}
    -time, where
    \begin{align*}
        K&= O\left(
        \frac{B_1}{\Delta_{\min}} + \log{\left(\frac{1}{\varepsilon^*}\right)}
        \right), \\
        M&=  O\left( \frac{1}{\varepsilon^*} \left(\frac{B_2}{\Delta_{\min}}+\frac{B_1^2}{\Delta_{\min}^2}\right) \right), \\
        s&= O\left( m2^{O(R^d)} \right), \\
        R&=O\left(\frac {1}{\Delta_{\min}}
          \log^2\left(\frac{B_1}{\Delta_{\min} \varepsilon^*} \right)\right),\\
        \varepsilon^* &= \Theta \left(\gamma-\frac12+\frac{\Delta_{\min}}{8A_0}\right),
    \end{align*}
    and $A_0$ is a known upper-bound of $\|H(0)\|$.
    For example, if we suppose that $\bm{sq}(\ket{u})=\bm{q}(\ket{u})=\bm{q}(D_R)=\mathrm{poly}(n)$, $A_0=m=\mathrm{poly}(n)$, $\Delta_{\min}=\Theta(1)$, $B_1=O(1)$, and $B_2=\mathrm{poly}(n)$,
    the algorithm runs in polynomial time for $\gamma\geq 1/2 + \Omega(1)$ and in quasi-polynomial time for $\gamma \geq 1/2$.
\end{thm}
\begin{proof}
    Let 
    \begin{equation}\label{eq:weak-energy-filter}
        G_0=\frac34I-\frac{H(0)}{4A_0},
        \qquad\frac12I\leq G_0\leq I ,
    \end{equation}
    be a degree-1 spectral filter of $H(0)$, where $A_0$ is a known upper-bound of $\|H(0)\|$.
    We estimate $z=\bra uU_D(2\pi,0)G_0\ket u$ instead of $\braket{u|U_D(2\pi,0)|u}$.
    Write $p=|\braket{u|\psi_0(0)}|^2$, and let $E_0,E_1$ be the two lowest eigenvalues of $H(0)$. 
    Expanding the guiding state as $\ket{u}=\sqrt{p}\ket{\psi_0(0)}+\sqrt{1-p}\ket{u_\perp}$, we have
    \begin{equation*}
        z = pe^{i\theta_B} \left(\frac{3}{4} - \frac{E_0}{4A_0} \right)
        + (1-p) \braket{u_\perp|U_D(2\pi,0)G_0|u_\perp} .
    \end{equation*}
    Since $e^{i \theta_B}\in\{+1,-1 \}$, 
    \begin{align*}
        e^{i\theta_B} \mathrm{Re}(z) &\geq p \left(\frac{3}{4} - \frac{E_0}{4A_0} \right)
        - (1-p) \left(\frac{3}{4} - \frac{E_1}{4A_0} \right) \\
        &= (2p-1)\left(\frac34-\frac{E_1}{4A_0}\right)
          +\frac{p(E_1-E_0)}{4A_0}\notag\\
        &\geq p-\frac12+\frac{\Delta_{\min}}{8A_0}
 \geq\gamma-\frac12+\frac{\Delta_{\min}}{8A_0}.
    \end{align*}
    Therefore we only have to determine the sign of $\mathrm{Re}(z)$.
    This can be achieved by estimating $\mathrm{Re}(z)$ with an additive error $\varepsilon^* \leq O\left(\gamma-\frac12+\frac{\Delta_{\min}}{8A_0}\right)$.
    The runtime of this classical algorithm is obtained by substituting $\varepsilon\leftarrow \varepsilon^*$ and $\bm{q}(\ket{u})\leftarrow \bm{q}(G_0\ket{u})$ in Theorem~\ref{thm:classical-BPE}.
    $H(0)$ is a $k$-local Hamiltonian and has at most $m2^k$ nonzero entries per row, and thus $\bm{q}(G_0\ket{u})=O(m2^k\bm{q}(\ket{u}))$.
\end{proof}

\section{Conclusion}\label{sec:conclusion}
We have established hardness results and classical algorithms for quantized Berry phase estimation in quantum many-body systems with access to a guiding state. 
Distinguishing Berry phases of $0$ and $\pi$ is $\BQP$-hard even for Hamiltonians on a two-dimensional square lattice drawn from any fixed non-2SLD interaction set, with $B_1=O(1)$ and an inverse-polynomial spectral gap. 
To show this, we encode quantum computation into a symmetry-quantized topological invariant of a Hamiltonian loop.
The encoded value is thus robust to symmetry-preserving deformations that do not close the gap.
In contrast, for geometrically local Hamiltonians on fixed-dimensional lattices, we obtain polynomial-time classical algorithms when the spectral gap, $B_1$, and target precision are constant and the guiding state is sufficiently accurate.
We believe that our results provide further motivation for exploring the power of quantum computation in characterizing topological properties of quantum many-body systems.

Several open problems remain. 
First, what is the complexity of Berry phase estimation when the spectral gap remains constant but $B_1$ is allowed to grow polynomially in $n$, even at constant precision and with an accurate guiding state? 
Does this regime remain classically tractable, or can sufficiently large variation revive $\mathsf{BQP}$-hardness without a small gap?
This question is closely related to the computational power of constant-gap adiabatic evolution.
While constant-gap adiabatic evolution is classically simulable in one dimension under suitable assumptions~\cite{hastings2009quantum}, universal quantum computation can be realized with a constant gap using adiabatic gate teleportation, at the price of a degenerate ground space~\cite{bacon2009adiabatic}.

Second, can our classical algorithms be extended to guiding states with smaller overlap? 
In particular, it remains open whether efficient classical estimation is possible with an arbitrary nonzero constant, or even inverse-polynomial, overlap with the initial ground state.

Third, what is the complexity of constant-precision Berry phase estimation without a guiding state?
Previous work~\cite{hayakawa2025computational} establishes $\mathsf{UQMA}\cap\mathsf{co}\text{-}\mathsf{UQMA}$-completeness at inverse-polynomial precision when an energy threshold that separates the ground energy and first excited energy is given.
It is essential to clarify how this characterization changes under exact phase quantization or a constant spectral gap.

The setting without guiding states also motivates quantifying quantum advantage.
In the context of ground energy estimation, recent works provide algorithms beating the natural Grover bound for a constant precision~\cite{buhrman2025beating} and reveal fine-grained lower bounds that prohibit such an improvement for a high-precision~\cite{chia2026fine}.
These results suggest analogous questions for Berry phase estimation.
Can Berry phase estimation be performed faster than the complexity of $O(2^{n/2})$, and can fine-grained upper and lower bounds characterize how the achievable speedup depends on the phase precision and spectral gap?

Finally, our results are related to recent work establishing quantum hardness of recognizing phases of matter from copies of an unknown quantum state under cryptographic assumptions~\cite{schuster2025hardness}.
Both settings assume access to a ground state, or a state closely related to it, and ask how difficult it is to extract information characterizing the underlying many-body system.
An important difference is that the hard instances of Ref.~\cite{schuster2025hardness} require parent Hamiltonians whose locality grows logarithmically, whereas our hardness holds for constant-local Hamiltonians.
Another distinction is that the Berry phase need not by itself identify the phase of its ground state.
It would be interesting to understand when the complexity of computing such explicit invariants can be related to the complexity of recognizing the underlying phase, particularly for constant-local Hamiltonians.

\begin{acknowledgments}
This work is supported by MEXT Quantum Leap Flagship Program (MEXT Q-LEAP) Grant No. JPMXS0120319794, JST COI-NEXT Grant No. JPMJPF2014, and JST CREST JPMJCR24I3.
K.S. is supported by JST SPRING Grant No. JPMJSP2138.
\end{acknowledgments}

\section*{Statement of AI use}
GPT-5.6 and GPT-6 were used to explore approaches to the proof of Theorem 2, Theorem 3 and Theorem 4, to check intermediate arguments, and to refine the text.
The final proof was written and independently verified by the authors.


\bibliography{ref}

\begin{thebibliography}{56}%
\makeatletter
\providecommand \@ifxundefined [1]{%
 \@ifx{#1\undefined}
}%
\providecommand \@ifnum [1]{%
 \ifnum #1\expandafter \@firstoftwo
 \else \expandafter \@secondoftwo
 \fi
}%
\providecommand \@ifx [1]{%
 \ifx #1\expandafter \@firstoftwo
 \else \expandafter \@secondoftwo
 \fi
}%
\providecommand \natexlab [1]{#1}%
\providecommand \enquote  [1]{``#1''}%
\providecommand \bibnamefont  [1]{#1}%
\providecommand \bibfnamefont [1]{#1}%
\providecommand \citenamefont [1]{#1}%
\providecommand \href@noop [0]{\@secondoftwo}%
\providecommand \href [0]{\begingroup \@sanitize@url \@href}%
\providecommand \@href[1]{\@@startlink{#1}\@@href}%
\providecommand \@@href[1]{\endgroup#1\@@endlink}%
\providecommand \@sanitize@url [0]{\catcode `\\12\catcode `\$12\catcode `\&12\catcode `\#12\catcode `\^12\catcode `\_12\catcode `\%12\relax}%
\providecommand \@@startlink[1]{}%
\providecommand \@@endlink[0]{}%
\providecommand \url  [0]{\begingroup\@sanitize@url \@url }%
\providecommand \@url [1]{\endgroup\@href {#1}{\urlprefix }}%
\providecommand \urlprefix  [0]{URL }%
\providecommand \Eprint [0]{\href }%
\providecommand \doibase [0]{https://doi.org/}%
\providecommand \selectlanguage [0]{\@gobble}%
\providecommand \bibinfo  [0]{\@secondoftwo}%
\providecommand \bibfield  [0]{\@secondoftwo}%
\providecommand \translation [1]{[#1]}%
\providecommand \BibitemOpen [0]{}%
\providecommand \bibitemStop [0]{}%
\providecommand \bibitemNoStop [0]{.\EOS\space}%
\providecommand \EOS [0]{\spacefactor3000\relax}%
\providecommand \BibitemShut  [1]{\csname bibitem#1\endcsname}%
\let\auto@bib@innerbib\@empty
\bibitem [{\citenamefont {Landau}\ and\ \citenamefont {Lifshitz}(2013)}]{landau2013course}%
  \BibitemOpen
  \bibfield  {author} {\bibinfo {author} {\bibfnamefont {L.~D.}\ \bibnamefont {Landau}}\ and\ \bibinfo {author} {\bibfnamefont {E.~M.}\ \bibnamefont {Lifshitz}},\ }\href@noop {} {\emph {\bibinfo {title} {Course of theoretical physics}}}\ (\bibinfo  {publisher} {Elsevier},\ \bibinfo {year} {2013})\BibitemShut {NoStop}%
\bibitem [{\citenamefont {Wen}(2017)}]{wen2017colloquium}%
  \BibitemOpen
  \bibfield  {author} {\bibinfo {author} {\bibfnamefont {X.-G.}\ \bibnamefont {Wen}},\ }\bibfield  {title} {\bibinfo {title} {Colloquium: Zoo of quantum-topological phases of matter},\ }\href {https://doi.org/10.1103/RevModPhys.89.041004} {\bibfield  {journal} {\bibinfo  {journal} {Reviews of Modern Physics}\ }\textbf {\bibinfo {volume} {89}},\ \bibinfo {pages} {041004} (\bibinfo {year} {2017})}\BibitemShut {NoStop}%
\bibitem [{\citenamefont {Senthil}(2015)}]{senthil2015symmetry}%
  \BibitemOpen
  \bibfield  {author} {\bibinfo {author} {\bibfnamefont {T.}~\bibnamefont {Senthil}},\ }\bibfield  {title} {\bibinfo {title} {Symmetry-protected topological phases of quantum matter},\ }\href {https://doi.org/https://doi.org/10.1146/annurev-conmatphys-031214-014740} {\bibfield  {journal} {\bibinfo  {journal} {Annu. Rev. Condens. Matter Phys.}\ }\textbf {\bibinfo {volume} {6}},\ \bibinfo {pages} {299} (\bibinfo {year} {2015})}\BibitemShut {NoStop}%
\bibitem [{\citenamefont {Chiu}\ \emph {et~al.}(2016)\citenamefont {Chiu}, \citenamefont {Teo}, \citenamefont {Schnyder},\ and\ \citenamefont {Ryu}}]{chiu2016classification}%
  \BibitemOpen
  \bibfield  {author} {\bibinfo {author} {\bibfnamefont {C.-K.}\ \bibnamefont {Chiu}}, \bibinfo {author} {\bibfnamefont {J.~C.}\ \bibnamefont {Teo}}, \bibinfo {author} {\bibfnamefont {A.~P.}\ \bibnamefont {Schnyder}},\ and\ \bibinfo {author} {\bibfnamefont {S.}~\bibnamefont {Ryu}},\ }\bibfield  {title} {\bibinfo {title} {Classification of topological quantum matter with symmetries},\ }\href@noop {} {\bibfield  {journal} {\bibinfo  {journal} {Reviews of Modern Physics}\ }\textbf {\bibinfo {volume} {88}},\ \bibinfo {pages} {035005} (\bibinfo {year} {2016})}\BibitemShut {NoStop}%
\bibitem [{\citenamefont {Chen}\ \emph {et~al.}(2010)\citenamefont {Chen}, \citenamefont {Gu},\ and\ \citenamefont {Wen}}]{chen2010local}%
  \BibitemOpen
  \bibfield  {author} {\bibinfo {author} {\bibfnamefont {X.}~\bibnamefont {Chen}}, \bibinfo {author} {\bibfnamefont {Z.-C.}\ \bibnamefont {Gu}},\ and\ \bibinfo {author} {\bibfnamefont {X.-G.}\ \bibnamefont {Wen}},\ }\bibfield  {title} {\bibinfo {title} {Local unitary transformation, long-range quantum entanglement, wave function renormalization, and topological order},\ }\href {https://doi.org/10.1103/PhysRevB.82.155138} {\bibfield  {journal} {\bibinfo  {journal} {Physical Review B—Condensed Matter and Materials Physics}\ }\textbf {\bibinfo {volume} {82}},\ \bibinfo {pages} {155138} (\bibinfo {year} {2010})}\BibitemShut {NoStop}%
\bibitem [{\citenamefont {Wen}(2013)}]{wen2013topological}%
  \BibitemOpen
  \bibfield  {author} {\bibinfo {author} {\bibfnamefont {X.-G.}\ \bibnamefont {Wen}},\ }\bibfield  {title} {\bibinfo {title} {Topological order: From long-range entangled quantum matter to a unified origin of light and electrons},\ }\href {https://doi.org/https://doi.org/10.1155/2013/198710} {\bibfield  {journal} {\bibinfo  {journal} {International Scholarly Research Notices}\ }\textbf {\bibinfo {volume} {2013}},\ \bibinfo {pages} {198710} (\bibinfo {year} {2013})}\BibitemShut {NoStop}%
\bibitem [{\citenamefont {Berry}(1984)}]{berry1984quantal}%
  \BibitemOpen
  \bibfield  {author} {\bibinfo {author} {\bibfnamefont {M.~V.}\ \bibnamefont {Berry}},\ }\bibfield  {title} {\bibinfo {title} {Quantal phase factors accompanying adiabatic changes},\ }\href {https://doi.org/https://doi.org/10.1098/rspa.1984.0023} {\bibfield  {journal} {\bibinfo  {journal} {Proceedings of the Royal Society of London. A. Mathematical and Physical Sciences}\ }\textbf {\bibinfo {volume} {392}},\ \bibinfo {pages} {45} (\bibinfo {year} {1984})}\BibitemShut {NoStop}%
\bibitem [{\citenamefont {King-Smith}\ and\ \citenamefont {Vanderbilt}(1993)}]{king1993theory}%
  \BibitemOpen
  \bibfield  {author} {\bibinfo {author} {\bibfnamefont {R.}~\bibnamefont {King-Smith}}\ and\ \bibinfo {author} {\bibfnamefont {D.}~\bibnamefont {Vanderbilt}},\ }\bibfield  {title} {\bibinfo {title} {Theory of polarization of crystalline solids},\ }\href {https://doi.org/10.1103/PhysRevB.47.1651} {\bibfield  {journal} {\bibinfo  {journal} {Physical Review B}\ }\textbf {\bibinfo {volume} {47}},\ \bibinfo {pages} {1651} (\bibinfo {year} {1993})}\BibitemShut {NoStop}%
\bibitem [{\citenamefont {Resta}(1994)}]{resta1994macroscopic}%
  \BibitemOpen
  \bibfield  {author} {\bibinfo {author} {\bibfnamefont {R.}~\bibnamefont {Resta}},\ }\bibfield  {title} {\bibinfo {title} {Macroscopic polarization in crystalline dielectrics: the geometric phase approach},\ }\href {https://doi.org/10.1103/RevModPhys.66.899} {\bibfield  {journal} {\bibinfo  {journal} {Reviews of modern physics}\ }\textbf {\bibinfo {volume} {66}},\ \bibinfo {pages} {899} (\bibinfo {year} {1994})}\BibitemShut {NoStop}%
\bibitem [{\citenamefont {Thouless}\ \emph {et~al.}(1982)\citenamefont {Thouless}, \citenamefont {Kohmoto}, \citenamefont {Nightingale},\ and\ \citenamefont {den Nijs}}]{thouless1982quantized}%
  \BibitemOpen
  \bibfield  {author} {\bibinfo {author} {\bibfnamefont {D.~J.}\ \bibnamefont {Thouless}}, \bibinfo {author} {\bibfnamefont {M.}~\bibnamefont {Kohmoto}}, \bibinfo {author} {\bibfnamefont {M.~P.}\ \bibnamefont {Nightingale}},\ and\ \bibinfo {author} {\bibfnamefont {M.}~\bibnamefont {den Nijs}},\ }\bibfield  {title} {\bibinfo {title} {Quantized hall conductance in a two-dimensional periodic potential},\ }\href {https://doi.org/10.1103/PhysRevLett.49.405} {\bibfield  {journal} {\bibinfo  {journal} {Physical review letters}\ }\textbf {\bibinfo {volume} {49}},\ \bibinfo {pages} {405} (\bibinfo {year} {1982})}\BibitemShut {NoStop}%
\bibitem [{\citenamefont {Niu}\ \emph {et~al.}(1985)\citenamefont {Niu}, \citenamefont {Thouless},\ and\ \citenamefont {Wu}}]{niu1985quantized}%
  \BibitemOpen
  \bibfield  {author} {\bibinfo {author} {\bibfnamefont {Q.}~\bibnamefont {Niu}}, \bibinfo {author} {\bibfnamefont {D.~J.}\ \bibnamefont {Thouless}},\ and\ \bibinfo {author} {\bibfnamefont {Y.-S.}\ \bibnamefont {Wu}},\ }\bibfield  {title} {\bibinfo {title} {Quantized hall conductance as a topological invariant},\ }\href {https://doi.org/10.1103/PhysRevB.31.3372} {\bibfield  {journal} {\bibinfo  {journal} {Physical Review B}\ }\textbf {\bibinfo {volume} {31}},\ \bibinfo {pages} {3372} (\bibinfo {year} {1985})}\BibitemShut {NoStop}%
\bibitem [{\citenamefont {Zak}(1989)}]{zak1989berry}%
  \BibitemOpen
  \bibfield  {author} {\bibinfo {author} {\bibfnamefont {J.}~\bibnamefont {Zak}},\ }\bibfield  {title} {\bibinfo {title} {Berry’s phase for energy bands in solids},\ }\href {https://doi.org/10.1103/PhysRevLett.62.2747} {\bibfield  {journal} {\bibinfo  {journal} {Physical review letters}\ }\textbf {\bibinfo {volume} {62}},\ \bibinfo {pages} {2747} (\bibinfo {year} {1989})}\BibitemShut {NoStop}%
\bibitem [{\citenamefont {Hatsugai}(2006)}]{hatsugai2006quantized}%
  \BibitemOpen
  \bibfield  {author} {\bibinfo {author} {\bibfnamefont {Y.}~\bibnamefont {Hatsugai}},\ }\bibfield  {title} {\bibinfo {title} {Quantized berry phases as a local order parameter of a quantum liquid},\ }\href {https://doi.org/https://doi.org/10.1143/jpsj.75.123601} {\bibfield  {journal} {\bibinfo  {journal} {Journal of the Physical Society of Japan}\ }\textbf {\bibinfo {volume} {75}},\ \bibinfo {pages} {123601} (\bibinfo {year} {2006})}\BibitemShut {NoStop}%
\bibitem [{\citenamefont {Hirano}\ \emph {et~al.}(2008)\citenamefont {Hirano}, \citenamefont {Katsura},\ and\ \citenamefont {Hatsugai}}]{hirano2008topological}%
  \BibitemOpen
  \bibfield  {author} {\bibinfo {author} {\bibfnamefont {T.}~\bibnamefont {Hirano}}, \bibinfo {author} {\bibfnamefont {H.}~\bibnamefont {Katsura}},\ and\ \bibinfo {author} {\bibfnamefont {Y.}~\bibnamefont {Hatsugai}},\ }\bibfield  {title} {\bibinfo {title} {Topological classification of gapped spin chains: Quantized berry phase as a local order parameter},\ }\href {https://doi.org/10.1103/PhysRevB.77.094431} {\bibfield  {journal} {\bibinfo  {journal} {Physical Review B—Condensed Matter and Materials Physics}\ }\textbf {\bibinfo {volume} {77}},\ \bibinfo {pages} {094431} (\bibinfo {year} {2008})}\BibitemShut {NoStop}%
\bibitem [{\citenamefont {Araki}\ \emph {et~al.}(2020)\citenamefont {Araki}, \citenamefont {Mizoguchi},\ and\ \citenamefont {Hatsugai}}]{araki2020zq}%
  \BibitemOpen
  \bibfield  {author} {\bibinfo {author} {\bibfnamefont {H.}~\bibnamefont {Araki}}, \bibinfo {author} {\bibfnamefont {T.}~\bibnamefont {Mizoguchi}},\ and\ \bibinfo {author} {\bibfnamefont {Y.}~\bibnamefont {Hatsugai}},\ }\bibfield  {title} {\bibinfo {title} {Zq berry phase for higher-order symmetry-protected topological phases},\ }\href {https://doi.org/10.1103/PhysRevResearch.2.012009} {\bibfield  {journal} {\bibinfo  {journal} {Physical Review Research}\ }\textbf {\bibinfo {volume} {2}},\ \bibinfo {pages} {012009} (\bibinfo {year} {2020})}\BibitemShut {NoStop}%
\bibitem [{\citenamefont {Pollmann}\ and\ \citenamefont {Turner}(2012)}]{pollmann2012detection}%
  \BibitemOpen
  \bibfield  {author} {\bibinfo {author} {\bibfnamefont {F.}~\bibnamefont {Pollmann}}\ and\ \bibinfo {author} {\bibfnamefont {A.~M.}\ \bibnamefont {Turner}},\ }\bibfield  {title} {\bibinfo {title} {Detection of symmetry-protected topological phases in one dimension},\ }\href {https://doi.org/10.1103/PhysRevB.86.125441} {\bibfield  {journal} {\bibinfo  {journal} {Physical Review B—Condensed Matter and Materials Physics}\ }\textbf {\bibinfo {volume} {86}},\ \bibinfo {pages} {125441} (\bibinfo {year} {2012})}\BibitemShut {NoStop}%
\bibitem [{\citenamefont {Motoyama}\ and\ \citenamefont {Todo}(2013)}]{motoyama2013path}%
  \BibitemOpen
  \bibfield  {author} {\bibinfo {author} {\bibfnamefont {Y.}~\bibnamefont {Motoyama}}\ and\ \bibinfo {author} {\bibfnamefont {S.}~\bibnamefont {Todo}},\ }\bibfield  {title} {\bibinfo {title} {Path-integral monte carlo method for the local z 2 berry phase},\ }\href {https://doi.org/10.1103/PhysRevE.87.021301} {\bibfield  {journal} {\bibinfo  {journal} {Physical Review E—Statistical, Nonlinear, and Soft Matter Physics}\ }\textbf {\bibinfo {volume} {87}},\ \bibinfo {pages} {021301} (\bibinfo {year} {2013})}\BibitemShut {NoStop}%
\bibitem [{\citenamefont {Murta}\ \emph {et~al.}(2020)\citenamefont {Murta}, \citenamefont {Catarina},\ and\ \citenamefont {Fern{\'a}ndez-Rossier}}]{murta2020berry}%
  \BibitemOpen
  \bibfield  {author} {\bibinfo {author} {\bibfnamefont {B.}~\bibnamefont {Murta}}, \bibinfo {author} {\bibfnamefont {G.}~\bibnamefont {Catarina}},\ and\ \bibinfo {author} {\bibfnamefont {J.}~\bibnamefont {Fern{\'a}ndez-Rossier}},\ }\bibfield  {title} {\bibinfo {title} {Berry phase estimation in gate-based adiabatic quantum simulation},\ }\href {https://doi.org/10.1103/PhysRevA.101.020302} {\bibfield  {journal} {\bibinfo  {journal} {Physical Review A}\ }\textbf {\bibinfo {volume} {101}},\ \bibinfo {pages} {020302} (\bibinfo {year} {2020})}\BibitemShut {NoStop}%
\bibitem [{\citenamefont {Tamiya}\ \emph {et~al.}(2021)\citenamefont {Tamiya}, \citenamefont {Koh},\ and\ \citenamefont {Nakagawa}}]{tamiya2021calculating}%
  \BibitemOpen
  \bibfield  {author} {\bibinfo {author} {\bibfnamefont {S.}~\bibnamefont {Tamiya}}, \bibinfo {author} {\bibfnamefont {S.}~\bibnamefont {Koh}},\ and\ \bibinfo {author} {\bibfnamefont {Y.~O.}\ \bibnamefont {Nakagawa}},\ }\bibfield  {title} {\bibinfo {title} {Calculating nonadiabatic couplings and berry's phase by variational quantum eigensolvers},\ }\href {https://doi.org/10.1103/PhysRevResearch.3.023244} {\bibfield  {journal} {\bibinfo  {journal} {Physical Review Research}\ }\textbf {\bibinfo {volume} {3}},\ \bibinfo {pages} {023244} (\bibinfo {year} {2021})}\BibitemShut {NoStop}%
\bibitem [{\citenamefont {Kiumi}(2026)}]{kiumi2026adiabatic}%
  \BibitemOpen
  \bibfield  {author} {\bibinfo {author} {\bibfnamefont {C.}~\bibnamefont {Kiumi}},\ }\bibfield  {title} {\bibinfo {title} {Adiabatic error cancellation in berry phase estimation},\ }\href {https://arxiv.org/abs/2604.20952} {\bibfield  {journal} {\bibinfo  {journal} {arXiv preprint arXiv:2604.20952}\ } (\bibinfo {year} {2026})}\BibitemShut {NoStop}%
\bibitem [{\citenamefont {Hayakawa}\ \emph {et~al.}(2025)\citenamefont {Hayakawa}, \citenamefont {Sakamoto},\ and\ \citenamefont {Kiumi}}]{hayakawa2025computational}%
  \BibitemOpen
  \bibfield  {author} {\bibinfo {author} {\bibfnamefont {R.}~\bibnamefont {Hayakawa}}, \bibinfo {author} {\bibfnamefont {K.}~\bibnamefont {Sakamoto}},\ and\ \bibinfo {author} {\bibfnamefont {C.}~\bibnamefont {Kiumi}},\ }\bibfield  {title} {\bibinfo {title} {Computational complexity of berry phase estimation in topological phases of matter},\ }\href {https://arxiv.org/abs/2509.13423} {\bibfield  {journal} {\bibinfo  {journal} {arXiv preprint arXiv:2509.13423}\ } (\bibinfo {year} {2025})}\BibitemShut {NoStop}%
\bibitem [{\citenamefont {Gharibian}\ and\ \citenamefont {Le~Gall}(2022)}]{gharibian2022dequantizing}%
  \BibitemOpen
  \bibfield  {author} {\bibinfo {author} {\bibfnamefont {S.}~\bibnamefont {Gharibian}}\ and\ \bibinfo {author} {\bibfnamefont {F.}~\bibnamefont {Le~Gall}},\ }\bibfield  {title} {\bibinfo {title} {Dequantizing the quantum singular value transformation: hardness and applications to quantum chemistry and the quantum pcp conjecture},\ }in\ \href {https://doi.org/10.1145/3519935.3519991} {\emph {\bibinfo {booktitle} {Proceedings of the 54th annual ACM SIGACT symposium on theory of computing}}}\ (\bibinfo {year} {2022})\ pp.\ \bibinfo {pages} {19--32}\BibitemShut {NoStop}%
\bibitem [{\citenamefont {Gall}(2024)}]{gall2024classical}%
  \BibitemOpen
  \bibfield  {author} {\bibinfo {author} {\bibfnamefont {F.~L.}\ \bibnamefont {Gall}},\ }\bibfield  {title} {\bibinfo {title} {Classical algorithms for constant approximation of the ground state energy of local hamiltonians},\ }\href {https://arxiv.org/abs/2410.21833} {\bibfield  {journal} {\bibinfo  {journal} {arXiv preprint arXiv:2410.21833}\ } (\bibinfo {year} {2024})}\BibitemShut {NoStop}%
\bibitem [{Note1()}]{Note1}%
  \BibitemOpen
  \bibinfo {note} {The term 2SLD stands for ''the 2-local parts of all interactions in the set are simultaneously locally diagonalizable''. The concept was originally introduced in Ref.~\cite {cubitt2016complexity}.}\BibitemShut {Stop}%
\bibitem [{\citenamefont {Su}\ \emph {et~al.}(1979)\citenamefont {Su}, \citenamefont {Schrieffer},\ and\ \citenamefont {Heeger}}]{su1979solitons}%
  \BibitemOpen
  \bibfield  {author} {\bibinfo {author} {\bibfnamefont {W.-P.}\ \bibnamefont {Su}}, \bibinfo {author} {\bibfnamefont {J.~R.}\ \bibnamefont {Schrieffer}},\ and\ \bibinfo {author} {\bibfnamefont {A.~J.}\ \bibnamefont {Heeger}},\ }\bibfield  {title} {\bibinfo {title} {Solitons in polyacetylene},\ }\href {https://doi.org/10.1103/PhysRevLett.42.1698} {\bibfield  {journal} {\bibinfo  {journal} {Physical review letters}\ }\textbf {\bibinfo {volume} {42}},\ \bibinfo {pages} {1698} (\bibinfo {year} {1979})}\BibitemShut {NoStop}%
\bibitem [{\citenamefont {Kitaev}\ \emph {et~al.}(2002)\citenamefont {Kitaev}, \citenamefont {Shen},\ and\ \citenamefont {Vyalyi}}]{kitaev2002classical}%
  \BibitemOpen
  \bibfield  {author} {\bibinfo {author} {\bibfnamefont {A.~Y.}\ \bibnamefont {Kitaev}}, \bibinfo {author} {\bibfnamefont {A.}~\bibnamefont {Shen}},\ and\ \bibinfo {author} {\bibfnamefont {M.~N.}\ \bibnamefont {Vyalyi}},\ }\href {https://doi.org/10.1090/gsm/047} {\emph {\bibinfo {title} {Classical and quantum computation}}}\ (\bibinfo  {publisher} {American Mathematical Society},\ \bibinfo {year} {2002})\BibitemShut {NoStop}%
\bibitem [{\citenamefont {Feynman}(1986)}]{feynman1986quantum}%
  \BibitemOpen
  \bibfield  {author} {\bibinfo {author} {\bibfnamefont {R.~P.}\ \bibnamefont {Feynman}},\ }\bibfield  {title} {\bibinfo {title} {Quantum mechanical computers.},\ }\href {https://doi.org/10.1007/BF01886518} {\bibfield  {journal} {\bibinfo  {journal} {Found. Phys.}\ }\textbf {\bibinfo {volume} {16}},\ \bibinfo {pages} {507} (\bibinfo {year} {1986})}\BibitemShut {NoStop}%
\bibitem [{\citenamefont {Bravyi}\ \emph {et~al.}(2011)\citenamefont {Bravyi}, \citenamefont {DiVincenzo},\ and\ \citenamefont {Loss}}]{bravyi2011schrieffer}%
  \BibitemOpen
  \bibfield  {author} {\bibinfo {author} {\bibfnamefont {S.}~\bibnamefont {Bravyi}}, \bibinfo {author} {\bibfnamefont {D.~P.}\ \bibnamefont {DiVincenzo}},\ and\ \bibinfo {author} {\bibfnamefont {D.}~\bibnamefont {Loss}},\ }\bibfield  {title} {\bibinfo {title} {Schrieffer--wolff transformation for quantum many-body systems},\ }\href {https://doi.org/https://doi.org/10.1016/j.aop.2011.06.004} {\bibfield  {journal} {\bibinfo  {journal} {Annals of physics}\ }\textbf {\bibinfo {volume} {326}},\ \bibinfo {pages} {2793} (\bibinfo {year} {2011})}\BibitemShut {NoStop}%
\bibitem [{\citenamefont {Bravyi}\ and\ \citenamefont {Hastings}(2017)}]{bravyi2017complexity}%
  \BibitemOpen
  \bibfield  {author} {\bibinfo {author} {\bibfnamefont {S.}~\bibnamefont {Bravyi}}\ and\ \bibinfo {author} {\bibfnamefont {M.}~\bibnamefont {Hastings}},\ }\bibfield  {title} {\bibinfo {title} {On complexity of the quantum ising model},\ }\href {https://doi.org/https://doi.org/10.1007/s00220-016-2787-4} {\bibfield  {journal} {\bibinfo  {journal} {Communications in Mathematical Physics}\ }\textbf {\bibinfo {volume} {349}},\ \bibinfo {pages} {1} (\bibinfo {year} {2017})}\BibitemShut {NoStop}%
\bibitem [{\citenamefont {Aharonov}\ \emph {et~al.}(2008)\citenamefont {Aharonov}, \citenamefont {Van~Dam}, \citenamefont {Kempe}, \citenamefont {Landau}, \citenamefont {Lloyd},\ and\ \citenamefont {Regev}}]{aharonov2008adiabatic}%
  \BibitemOpen
  \bibfield  {author} {\bibinfo {author} {\bibfnamefont {D.}~\bibnamefont {Aharonov}}, \bibinfo {author} {\bibfnamefont {W.}~\bibnamefont {Van~Dam}}, \bibinfo {author} {\bibfnamefont {J.}~\bibnamefont {Kempe}}, \bibinfo {author} {\bibfnamefont {Z.}~\bibnamefont {Landau}}, \bibinfo {author} {\bibfnamefont {S.}~\bibnamefont {Lloyd}},\ and\ \bibinfo {author} {\bibfnamefont {O.}~\bibnamefont {Regev}},\ }\bibfield  {title} {\bibinfo {title} {Adiabatic quantum computation is equivalent to standard quantum computation},\ }\href {https://doi.org/https://doi.org/10.1137/080734479} {\bibfield  {journal} {\bibinfo  {journal} {SIAM review}\ }\textbf {\bibinfo {volume} {50}},\ \bibinfo {pages} {755} (\bibinfo {year} {2008})}\BibitemShut {NoStop}%
\bibitem [{\citenamefont {Cade}\ \emph {et~al.}(2023)\citenamefont {Cade}, \citenamefont {Folkertsma}, \citenamefont {Gharibian}, \citenamefont {Hayakawa}, \citenamefont {Le~Gall}, \citenamefont {Morimae},\ and\ \citenamefont {Weggemans}}]{cade2022improved}%
  \BibitemOpen
  \bibfield  {author} {\bibinfo {author} {\bibfnamefont {C.}~\bibnamefont {Cade}}, \bibinfo {author} {\bibfnamefont {M.}~\bibnamefont {Folkertsma}}, \bibinfo {author} {\bibfnamefont {S.}~\bibnamefont {Gharibian}}, \bibinfo {author} {\bibfnamefont {R.}~\bibnamefont {Hayakawa}}, \bibinfo {author} {\bibfnamefont {F.}~\bibnamefont {Le~Gall}}, \bibinfo {author} {\bibfnamefont {T.}~\bibnamefont {Morimae}},\ and\ \bibinfo {author} {\bibfnamefont {J.}~\bibnamefont {Weggemans}},\ }\bibfield  {title} {\bibinfo {title} {{Improved Hardness Results for the Guided Local Hamiltonian Problem}},\ }in\ \href {https://doi.org/10.4230/LIPIcs.ICALP.2023.32} {\emph {\bibinfo {booktitle} {50th International Colloquium on Automata, Languages, and Programming (ICALP 2023)}}},\ \bibinfo {series} {Leibniz International Proceedings in Informatics (LIPIcs)}, Vol.\ \bibinfo {volume} {261},\ \bibinfo {editor} {edited by\ \bibinfo {editor} {\bibfnamefont {K.}~\bibnamefont {Etessami}}, \bibinfo {editor} {\bibfnamefont {U.}~\bibnamefont
  {Feige}},\ and\ \bibinfo {editor} {\bibfnamefont {G.}~\bibnamefont {Puppis}}}\ (\bibinfo  {publisher} {Schloss Dagstuhl -- Leibniz-Zentrum f{\"u}r Informatik},\ \bibinfo {address} {Dagstuhl, Germany},\ \bibinfo {year} {2023})\ pp.\ \bibinfo {pages} {32:1--32:19}\BibitemShut {NoStop}%
\bibitem [{Note2()}]{Note2}%
  \BibitemOpen
  \bibinfo {note} {Very recently, independent work by Waite~\cite {waite2026computational} established BQP-completeness of guided Berry phase estimation at inverse-polynomial precision for parameterized 2-local qubit Hamiltonians, including weighted Heisenberg interactions on two-dimensional square and triangular lattices. Our work instead focuses on symmetry-quantized Berry phases taking exactly \(0\) or \(\pi \), for which we establish hardness already at constant precision. Moreover, the exact symmetry substantially simplifies the physical reduction by reducing Berry-phase preservation to the preservation of discrete symmetry eigenvalues.}\BibitemShut {Stop}%
\bibitem [{\citenamefont {Cubitt}\ \emph {et~al.}(2018)\citenamefont {Cubitt}, \citenamefont {Montanaro},\ and\ \citenamefont {Piddock}}]{cubitt2018universal}%
  \BibitemOpen
  \bibfield  {author} {\bibinfo {author} {\bibfnamefont {T.~S.}\ \bibnamefont {Cubitt}}, \bibinfo {author} {\bibfnamefont {A.}~\bibnamefont {Montanaro}},\ and\ \bibinfo {author} {\bibfnamefont {S.}~\bibnamefont {Piddock}},\ }\bibfield  {title} {\bibinfo {title} {Universal quantum hamiltonians},\ }\href {https://doi.org/https://doi.org/10.1073/pnas.1804949115} {\bibfield  {journal} {\bibinfo  {journal} {Proceedings of the National Academy of Sciences}\ }\textbf {\bibinfo {volume} {115}},\ \bibinfo {pages} {9497} (\bibinfo {year} {2018})}\BibitemShut {NoStop}%
\bibitem [{\citenamefont {Hastings}(2010{\natexlab{a}})}]{hastings2010locality}%
  \BibitemOpen
  \bibfield  {author} {\bibinfo {author} {\bibfnamefont {M.~B.}\ \bibnamefont {Hastings}},\ }\href {https://arxiv.org/abs/1008.5137} {\bibinfo {title} {Locality in quantum systems}} (\bibinfo {year} {2010}{\natexlab{a}}),\ \Eprint {https://arxiv.org/abs/1008.5137} {arXiv:1008.5137 [math-ph]} \BibitemShut {NoStop}%
\bibitem [{Note3()}]{Note3}%
  \BibitemOpen
  \bibinfo {note} {See Refs.~\cite {wu2024classical, zhang2024dequantized} for works considering other normalization and promise.}\BibitemShut {Stop}%
\bibitem [{\citenamefont {Hughes}\ \emph {et~al.}(2011)\citenamefont {Hughes}, \citenamefont {Prodan},\ and\ \citenamefont {Bernevig}}]{hughes2011inversion}%
  \BibitemOpen
  \bibfield  {author} {\bibinfo {author} {\bibfnamefont {T.~L.}\ \bibnamefont {Hughes}}, \bibinfo {author} {\bibfnamefont {E.}~\bibnamefont {Prodan}},\ and\ \bibinfo {author} {\bibfnamefont {B.~A.}\ \bibnamefont {Bernevig}},\ }\bibfield  {title} {\bibinfo {title} {Inversion-symmetric topological insulators},\ }\href {https://doi.org/10.1103/PhysRevB.83.245132} {\bibfield  {journal} {\bibinfo  {journal} {Physical Review B—Condensed Matter and Materials Physics}\ }\textbf {\bibinfo {volume} {83}},\ \bibinfo {pages} {245132} (\bibinfo {year} {2011})}\BibitemShut {NoStop}%
\bibitem [{\citenamefont {Gharibian}\ and\ \citenamefont {Kempe}(2012)}]{gharibian2012hardness}%
  \BibitemOpen
  \bibfield  {author} {\bibinfo {author} {\bibfnamefont {S.}~\bibnamefont {Gharibian}}\ and\ \bibinfo {author} {\bibfnamefont {J.}~\bibnamefont {Kempe}},\ }\bibfield  {title} {\bibinfo {title} {Hardness of approximation for quantum problems},\ }in\ \href {https://doi.org/https://doi.org/10.1007/978-3-642-31594-7_33} {\emph {\bibinfo {booktitle} {International Colloquium on Automata, Languages, and Programming}}}\ (\bibinfo {organization} {Springer},\ \bibinfo {year} {2012})\ pp.\ \bibinfo {pages} {387--398}\BibitemShut {NoStop}%
\bibitem [{\citenamefont {Jansen}\ \emph {et~al.}(2007)\citenamefont {Jansen}, \citenamefont {Ruskai},\ and\ \citenamefont {Seiler}}]{jansen2007bounds}%
  \BibitemOpen
  \bibfield  {author} {\bibinfo {author} {\bibfnamefont {S.}~\bibnamefont {Jansen}}, \bibinfo {author} {\bibfnamefont {M.-B.}\ \bibnamefont {Ruskai}},\ and\ \bibinfo {author} {\bibfnamefont {R.}~\bibnamefont {Seiler}},\ }\bibfield  {title} {\bibinfo {title} {Bounds for the adiabatic approximation with applications to quantum computation},\ }\bibfield  {journal} {\bibinfo  {journal} {Journal of Mathematical Physics}\ }\textbf {\bibinfo {volume} {48}},\ \href {https://doi.org/https://doi.org/10.1063/1.2798382} {https://doi.org/10.1063/1.2798382} (\bibinfo {year} {2007})\BibitemShut {NoStop}%
\bibitem [{\citenamefont {Zhou}\ and\ \citenamefont {Aharonov}(2026)}]{zhou2026universal}%
  \BibitemOpen
  \bibfield  {author} {\bibinfo {author} {\bibfnamefont {L.}~\bibnamefont {Zhou}}\ and\ \bibinfo {author} {\bibfnamefont {D.}~\bibnamefont {Aharonov}},\ }\bibfield  {title} {\bibinfo {title} {Universal hamiltonian simulators in one and two dimensions},\ }\bibfield  {journal} {\bibinfo  {journal} {Nature Communications}\ }\href {https://doi.org/https://doi.org/10.1038/s41467-026-71686-4} {https://doi.org/10.1038/s41467-026-71686-4} (\bibinfo {year} {2026})\BibitemShut {NoStop}%
\bibitem [{\citenamefont {Shi}(2003)}]{shi2002both}%
  \BibitemOpen
  \bibfield  {author} {\bibinfo {author} {\bibfnamefont {Y.}~\bibnamefont {Shi}},\ }\bibfield  {title} {\bibinfo {title} {Both toffoli and controlled-not need little help to do universal quantum computation},\ }\bibfield  {journal} {\bibinfo  {journal} {Quantum Information and Computation}\ }\href {https://doi.org/10.26421/qic3.1-7} {10.26421/qic3.1-7} (\bibinfo {year} {2003})\BibitemShut {NoStop}%
\bibitem [{\citenamefont {Aharonov}(2003)}]{aharonov2003simple}%
  \BibitemOpen
  \bibfield  {author} {\bibinfo {author} {\bibfnamefont {D.}~\bibnamefont {Aharonov}},\ }\bibfield  {title} {\bibinfo {title} {A simple proof that toffoli and hadamard are quantum universal},\ }\href {https://doi.org/10.48550/arXiv.quant-ph/0301040} {\bibfield  {journal} {\bibinfo  {journal} {arXiv preprint quant-ph/0301040}\ } (\bibinfo {year} {2003})},\ \Eprint {https://arxiv.org/abs/quant-ph/0301040} {quant-ph/0301040} \BibitemShut {NoStop}%
\bibitem [{\citenamefont {Oliveira}\ and\ \citenamefont {Terhal}(2008)}]{oliveira2005complexity}%
  \BibitemOpen
  \bibfield  {author} {\bibinfo {author} {\bibfnamefont {R.}~\bibnamefont {Oliveira}}\ and\ \bibinfo {author} {\bibfnamefont {B.~M.}\ \bibnamefont {Terhal}},\ }\bibfield  {title} {\bibinfo {title} {The complexity of quantum spin systems on a two-dimensional square lattice},\ }\href {https://doi.org/https://doi.org/10.26421/QIC8.10-2} {\bibfield  {journal} {\bibinfo  {journal} {Quantum Information and Computation}\ }\textbf {\bibinfo {volume} {8}},\ \bibinfo {pages} {900} (\bibinfo {year} {2008})}\BibitemShut {NoStop}%
\bibitem [{\citenamefont {Piddock}\ and\ \citenamefont {Montanaro}(2017)}]{piddock2015complexity}%
  \BibitemOpen
  \bibfield  {author} {\bibinfo {author} {\bibfnamefont {S.}~\bibnamefont {Piddock}}\ and\ \bibinfo {author} {\bibfnamefont {A.}~\bibnamefont {Montanaro}},\ }\bibfield  {title} {\bibinfo {title} {The complexity of antiferromagnetic interactions and 2d lattices},\ }\href {https://doi.org/https://doi.org/10.26421/QIC17.7-8-6} {\bibfield  {journal} {\bibinfo  {journal} {Quantum Info. Comput.}\ }\textbf {\bibinfo {volume} {17}},\ \bibinfo {pages} {636–672} (\bibinfo {year} {2017})}\BibitemShut {NoStop}%
\bibitem [{\citenamefont {Hastings}\ and\ \citenamefont {Wen}(2005)}]{hastings2005quasiadiabatic}%
  \BibitemOpen
  \bibfield  {author} {\bibinfo {author} {\bibfnamefont {M.~B.}\ \bibnamefont {Hastings}}\ and\ \bibinfo {author} {\bibfnamefont {X.-G.}\ \bibnamefont {Wen}},\ }\bibfield  {title} {\bibinfo {title} {Quasiadiabatic continuation of quantum states: The stability of topological ground-state degeneracy and emergent gauge invariance},\ }\href {https://doi.org/10.1103/PhysRevB.72.045141} {\bibfield  {journal} {\bibinfo  {journal} {Physical Review B—Condensed Matter and Materials Physics}\ }\textbf {\bibinfo {volume} {72}},\ \bibinfo {pages} {045141} (\bibinfo {year} {2005})}\BibitemShut {NoStop}%
\bibitem [{\citenamefont {Hastings}(2010{\natexlab{b}})}]{hastings2010quasi}%
  \BibitemOpen
  \bibfield  {author} {\bibinfo {author} {\bibfnamefont {M.~B.}\ \bibnamefont {Hastings}},\ }\bibfield  {title} {\bibinfo {title} {Quasi-adiabatic continuation for disordered systems: Applications to correlations, lieb-schultz-mattis, and hall conductance},\ }\href {https://arxiv.org/abs/1001.5280} {\bibfield  {journal} {\bibinfo  {journal} {arXiv preprint arXiv:1001.5280}\ } (\bibinfo {year} {2010}{\natexlab{b}})}\BibitemShut {NoStop}%
\bibitem [{\citenamefont {Harrow}\ and\ \citenamefont {Montanaro}(2017)}]{harrow2017extremal}%
  \BibitemOpen
  \bibfield  {author} {\bibinfo {author} {\bibfnamefont {A.~W.}\ \bibnamefont {Harrow}}\ and\ \bibinfo {author} {\bibfnamefont {A.}~\bibnamefont {Montanaro}},\ }\bibfield  {title} {\bibinfo {title} {Extremal eigenvalues of local hamiltonians},\ }\href {https://doi.org/10.22331/q-2017-04-25-6} {\bibfield  {journal} {\bibinfo  {journal} {Quantum}\ }\textbf {\bibinfo {volume} {1}},\ \bibinfo {pages} {6} (\bibinfo {year} {2017})}\BibitemShut {NoStop}%
\bibitem [{\citenamefont {Le~Gall}(2025)}]{gall2023robust}%
  \BibitemOpen
  \bibfield  {author} {\bibinfo {author} {\bibfnamefont {F.}~\bibnamefont {Le~Gall}},\ }\bibfield  {title} {\bibinfo {title} {Robust dequantization of the quantum singular value transformation and quantum machine learning algorithms},\ }\href {https://doi.org/10.1007/s00037-024-00262-3} {\bibfield  {journal} {\bibinfo  {journal} {computational complexity}\ }\textbf {\bibinfo {volume} {34}},\ \bibinfo {pages} {2} (\bibinfo {year} {2025})}\BibitemShut {NoStop}%
\bibitem [{\citenamefont {Hastings}(2009)}]{hastings2009quantum}%
  \BibitemOpen
  \bibfield  {author} {\bibinfo {author} {\bibfnamefont {M.~B.}\ \bibnamefont {Hastings}},\ }\bibfield  {title} {\bibinfo {title} {Quantum adiabatic computation with a constant gap is not useful in one dimension},\ }\href {https://doi.org/10.1103/PhysRevLett.103.050502} {\bibfield  {journal} {\bibinfo  {journal} {Physical review letters}\ }\textbf {\bibinfo {volume} {103}},\ \bibinfo {pages} {050502} (\bibinfo {year} {2009})}\BibitemShut {NoStop}%
\bibitem [{\citenamefont {Bacon}\ and\ \citenamefont {Flammia}(2009)}]{bacon2009adiabatic}%
  \BibitemOpen
  \bibfield  {author} {\bibinfo {author} {\bibfnamefont {D.}~\bibnamefont {Bacon}}\ and\ \bibinfo {author} {\bibfnamefont {S.~T.}\ \bibnamefont {Flammia}},\ }\bibfield  {title} {\bibinfo {title} {Adiabatic gate teleportation},\ }\href {https://doi.org/10.1103/PhysRevLett.103.120504} {\bibfield  {journal} {\bibinfo  {journal} {Physical review letters}\ }\textbf {\bibinfo {volume} {103}},\ \bibinfo {pages} {120504} (\bibinfo {year} {2009})}\BibitemShut {NoStop}%
\bibitem [{\citenamefont {Buhrman}\ \emph {et~al.}(2025)\citenamefont {Buhrman}, \citenamefont {Gharibian}, \citenamefont {Landau}, \citenamefont {Le~Gall}, \citenamefont {Schuch},\ and\ \citenamefont {Tamaki}}]{buhrman2025beating}%
  \BibitemOpen
  \bibfield  {author} {\bibinfo {author} {\bibfnamefont {H.}~\bibnamefont {Buhrman}}, \bibinfo {author} {\bibfnamefont {S.}~\bibnamefont {Gharibian}}, \bibinfo {author} {\bibfnamefont {Z.}~\bibnamefont {Landau}}, \bibinfo {author} {\bibfnamefont {F.}~\bibnamefont {Le~Gall}}, \bibinfo {author} {\bibfnamefont {N.}~\bibnamefont {Schuch}},\ and\ \bibinfo {author} {\bibfnamefont {S.}~\bibnamefont {Tamaki}},\ }\bibfield  {title} {\bibinfo {title} {Beating the natural grover bound for low-energy estimation and state preparation},\ }\href {https://doi.org/10.1103/29qw-bssx} {\bibfield  {journal} {\bibinfo  {journal} {Physical Review Letters}\ }\textbf {\bibinfo {volume} {135}},\ \bibinfo {pages} {030601} (\bibinfo {year} {2025})}\BibitemShut {NoStop}%
\bibitem [{\citenamefont {Chia}\ \emph {et~al.}(2026)\citenamefont {Chia}, \citenamefont {Hasegawa}, \citenamefont {Le~Gall},\ and\ \citenamefont {Shen}}]{chia2026fine}%
  \BibitemOpen
  \bibfield  {author} {\bibinfo {author} {\bibfnamefont {N.-H.}\ \bibnamefont {Chia}}, \bibinfo {author} {\bibfnamefont {A.}~\bibnamefont {Hasegawa}}, \bibinfo {author} {\bibfnamefont {F.}~\bibnamefont {Le~Gall}},\ and\ \bibinfo {author} {\bibfnamefont {Y.-C.}\ \bibnamefont {Shen}},\ }\bibfield  {title} {\bibinfo {title} {{Fine-Grained Complexity for Quantum Problems from Size-Preserving Circuit-To-Hamiltonian Constructions}},\ }in\ \href {https://doi.org/10.4230/LIPIcs.CCC.2026.12} {\emph {\bibinfo {booktitle} {41st Computational Complexity Conference (CCC 2026)}}},\ \bibinfo {series} {Leibniz International Proceedings in Informatics (LIPIcs)}, Vol.\ \bibinfo {volume} {383},\ \bibinfo {editor} {edited by\ \bibinfo {editor} {\bibfnamefont {D.}~\bibnamefont {Moshkovitz}}}\ (\bibinfo  {publisher} {Schloss Dagstuhl -- Leibniz-Zentrum f{\"u}r Informatik},\ \bibinfo {address} {Dagstuhl, Germany},\ \bibinfo {year} {2026})\ pp.\ \bibinfo {pages} {12:1--12:35}\BibitemShut {NoStop}%
\bibitem [{\citenamefont {Schuster}\ \emph {et~al.}(2025)\citenamefont {Schuster}, \citenamefont {Kufel}, \citenamefont {Yao},\ and\ \citenamefont {Huang}}]{schuster2025hardness}%
  \BibitemOpen
  \bibfield  {author} {\bibinfo {author} {\bibfnamefont {T.}~\bibnamefont {Schuster}}, \bibinfo {author} {\bibfnamefont {D.}~\bibnamefont {Kufel}}, \bibinfo {author} {\bibfnamefont {N.~Y.}\ \bibnamefont {Yao}},\ and\ \bibinfo {author} {\bibfnamefont {H.-Y.}\ \bibnamefont {Huang}},\ }\bibfield  {title} {\bibinfo {title} {Hardness of recognizing phases of matter},\ }\href {https://arxiv.org/abs/2510.08503} {\bibfield  {journal} {\bibinfo  {journal} {arXiv preprint arXiv:2510.08503}\ } (\bibinfo {year} {2025})}\BibitemShut {NoStop}%
\bibitem [{\citenamefont {Cubitt}\ and\ \citenamefont {Montanaro}(2016)}]{cubitt2016complexity}%
  \BibitemOpen
  \bibfield  {author} {\bibinfo {author} {\bibfnamefont {T.}~\bibnamefont {Cubitt}}\ and\ \bibinfo {author} {\bibfnamefont {A.}~\bibnamefont {Montanaro}},\ }\bibfield  {title} {\bibinfo {title} {Complexity classification of local hamiltonian problems},\ }\href {https://doi.org/https://doi.org/10.1137/140998287} {\bibfield  {journal} {\bibinfo  {journal} {SIAM Journal on Computing}\ }\textbf {\bibinfo {volume} {45}},\ \bibinfo {pages} {268} (\bibinfo {year} {2016})}\BibitemShut {NoStop}%
\bibitem [{\citenamefont {Waite}(2026)}]{waite2026computational}%
  \BibitemOpen
  \bibfield  {author} {\bibinfo {author} {\bibfnamefont {G.}~\bibnamefont {Waite}},\ }\href {https://arxiv.org/abs/2609.25929} {\bibinfo {title} {On the computational complexity of guided berry phase estimation}} (\bibinfo {year} {2026}),\ \Eprint {https://arxiv.org/abs/2609.25929} {arXiv:2609.25929 [quant-ph]} \BibitemShut {NoStop}%
\bibitem [{\citenamefont {Wu}\ \emph {et~al.}(2024)\citenamefont {Wu}, \citenamefont {Zhang}, \citenamefont {Wang},\ and\ \citenamefont {Yuan}}]{wu2024classical}%
  \BibitemOpen
  \bibfield  {author} {\bibinfo {author} {\bibfnamefont {Y.}~\bibnamefont {Wu}}, \bibinfo {author} {\bibfnamefont {Y.}~\bibnamefont {Zhang}}, \bibinfo {author} {\bibfnamefont {C.}~\bibnamefont {Wang}},\ and\ \bibinfo {author} {\bibfnamefont {X.}~\bibnamefont {Yuan}},\ }\bibfield  {title} {\bibinfo {title} {Classical algorithms for hamiltonian dynamics mean value and guided local hamiltonian problem},\ }\href {https://arxiv.org/abs/2409.04161} {\bibfield  {journal} {\bibinfo  {journal} {arXiv preprint arXiv:2409.04161}\ } (\bibinfo {year} {2024})}\BibitemShut {NoStop}%
\bibitem [{\citenamefont {Zhang}\ \emph {et~al.}(2024)\citenamefont {Zhang}, \citenamefont {Wu},\ and\ \citenamefont {Yuan}}]{zhang2024dequantized}%
  \BibitemOpen
  \bibfield  {author} {\bibinfo {author} {\bibfnamefont {Y.}~\bibnamefont {Zhang}}, \bibinfo {author} {\bibfnamefont {Y.}~\bibnamefont {Wu}},\ and\ \bibinfo {author} {\bibfnamefont {X.}~\bibnamefont {Yuan}},\ }\bibfield  {title} {\bibinfo {title} {A dequantized algorithm for the guided local hamiltonian problem},\ }\href {https://arxiv.org/abs/2411.16163} {\bibfield  {journal} {\bibinfo  {journal} {arXiv preprint arXiv:2411.16163}\ } (\bibinfo {year} {2024})}\BibitemShut {NoStop}%
\end{thebibliography}%


\appendix

\clearpage

\section{Berry phase from a parameter-reversing symmetry}\label{app:berry-symmetry}
We prove Eq.~\eqref{eq:berry-parity-formula} and Eq.~\eqref{eq:berry-sym}, which gives a relation between the Berry phase and the eigenvalues of the symmetry operator. 
Let $R$ obey Eq.~\eqref{eq:inversion-condition}, and then nondegeneracy implies
\begin{equation}
 R\ket{\psi_0(\theta)}=e^{i\beta(\theta)}\ket{\psi_0(-\theta)}.
\end{equation}
Differentiating this equality and considering that $R$ is independent of $\theta$ gives
\begin{equation}
 \mathcal A(\theta)+\mathcal A(-\theta)=-\dot\beta(\theta).
\end{equation}
Hence
\begin{align}
 \theta_B
 &=\int_{-\pi}^{\pi}\mathcal A(\theta)\,\mathrm d\theta
 =\int_0^\pi[\mathcal A(\theta)+\mathcal A(-\theta)]\,\mathrm d\theta\notag\\
 &=\beta(0)-\beta(\pi)\pmod{2\pi}.
\end{align}
At $\theta = 0$ and $\pi$, we have $R\ket{\psi_0(\theta)}=e^{i\beta(\theta)}\ket{\psi_0(\theta)}$. 
Since $R^2=I$, these eigenvalues are $\xi_0,\xi_\pi\in\{+1,-1\}$. Therefore $e^{i\theta_B}=\xi_0 \xi_\pi$.

\end{document}